\documentclass[sigconf]{acmart}

\def\cameraReady{}

\AtBeginDocument{}

\copyrightyear{2024}
\acmYear{2024}
\setcopyright{acmlicensed}\acmConference[CCS '24]{Proceedings of the 2024 ACM SIGSAC Conference on Computer and Communications Security}{October 14--18, 2024}{Salt Lake City, UT, USA}
\acmBooktitle{Proceedings of the 2024 ACM SIGSAC Conference on Computer and Communications Security (CCS '24), October 14--18, 2024, Salt Lake City, UT, USA}
\acmDOI{10.1145/3658644.3670286}
\acmISBN{979-8-4007-0636-3/24/10}

\usepackage{setspace}
\usepackage{xcolor}
\usepackage{amsfonts,amsmath,amsthm}
\usepackage{enumitem}
\setitemize{noitemsep,topsep=0pt,parsep=0pt,partopsep=0pt}
\setenumerate{noitemsep,topsep=0pt,parsep=0pt,partopsep=0pt}
\usepackage{cryptocode}
\usepackage{multirow, tabu, makecell, graphics}
\usepackage{algorithm}
\usepackage[noend]{algpseudocode}
\usepackage{amsmath}
\usepackage{pifont}
\usepackage{xspace}
\usepackage{bm}
\usepackage{enumitem}
\setlist{leftmargin=*}
\usepackage{hyperref}
\usepackage{listings}
\usepackage{cleveref}
\usepackage{pifont}
\usepackage[compact]{titlesec}

\ifdefined\cameraReady
  \newcommand{\sysname}{\textsc{Sui Lutris}\xspace}
\else
  \newcommand{\sysname}{\textsc{Lutris}\xspace}
\fi

\newcommand{\bragging}{
  As of May 2024, the Sui mainnet is operated by 107 geo-distributed heterogeneous validators and processes over 3.1 millions certificates a day (over 8.9 million ops/day through transaction blocks) over 383 epochs changes using the \sysname protocols. It stores over 143 million objects, owned by over 10.1 million addresses, defined by over 15,500 Move packages. On the peak throughput day (Jul.\ 27, 2023) it processed over 65 million certificates (most of which owned object transactions), the highest volume of any blockchain, and higher than the total number of transactions on all chains combined, on that day. \sysname secures over 700 million USD of value locked in Defi protocols, and a fully diluted market cap of over 12 billion USD of native cryptocurrency.
}
\newcommand{\extendedBragging}{
  As of May 2024, the Sui mainnet is operated by 107 geo-distributed heterogeneous validators, has processed over 1.2 billion certificates, at the rate of over 3.1 million certificates a day, in over 32 million checkpoints, and has undergone over 383 epoch changes. It stores over 143 million objects, owned by over 10.1 million addresses, defined by over 15,500 Move packages.
}

\newcommand{\objectid}{\textsf{ObjID}}

\newcommand{\objectkey}{\textsf{ObjKey}}
\newcommand{\objectversion}{\textsf{Version}}

\newcommand{\object}{\textsf{Obj}}
\newcommand{\cert}{\textsf{TxCert}}
\newcommand{\txsign}{\textsf{TxSign}}

\newcommand{\epoch}{\textsf{Epoch}}

\newcommand{\okey}[1]{\textit{key} ({#1})}
\newcommand{\version}[1]{\textit{version} ({#1})}
\newcommand{\initialversion}[1]{\textit{initial} ({#1})}

\newcommand{\owned}[1]{\textit{owned} ({#1})}
\newcommand{\oauth}[1]{\textit{owner} ({#1})}

\newcommand{\parent}[1]{\textit{creator} ({#1})}
\newcommand{\contents}[1]{\textit{contents} ({#1})}

\newcommand{\transaction}{  \textsf{Tx}}
\newcommand{\transactiondigest} {  \textsf{TxDigest}}
\newcommand{\effects}{  \textsf{Effects}}
\newcommand{\esign}{  \textsf{EffSign}}
\newcommand{\ecert}{  \textsf{EffCert}}

\newcommand{\status}{  \textsf{Status}}
\newcommand{\created}{  \textsf{Created}}
\newcommand{\mutated}{  \textsf{Mutated}}
\newcommand{\deleted}{  \textsf{Deleted}}
\newcommand{\events}{  \textsf{Events}}
\newcommand{\wrapped}{  \textsf{Wrapped}}
\newcommand{\unwrapped}{  \textsf{Unwrapped}}

\newcommand{\txepoch}[1]{\textit{epoch} ({#1})}
\newcommand{\txvalid}[1]{\textit{valid} ({#1})}
\newcommand{\txinputs}[1]{\textit{inputs} ({#1})}
\newcommand{\txsharedinputs}[1]{\textit{shared\_inputs} ({#1})}
\newcommand{\txreadonlyinputs}[1]{\textit{read\_only\_inputs} ({#1})}
\newcommand{\txdeps}[1]{\textit{dependencies} ({#1})}
\newcommand{\txecon}[1]{\textit{payment} ({#1})}
\newcommand{\txexec}[1]{\textit{exec} ({#1})}
\newcommand{\txdigest}[1]{\textit{digest} ({#1})}
\newcommand{\efftransaction}[1]{\textit{transaction} ({#1})}
\newcommand{\sign}[1]{\textit{sign} ({#1})}

\newcommand{\ownedlockdb}{  \textsf{OwnedLock}}
\newcommand{\sharedlockdb}{  \textsf{SharedLock}}
\newcommand{\nextsharedlockdb}{  \textsf{NextSharedLock}}
\newcommand{\objdb}{  \textsf{ObjDB}}
\newcommand{\certdb}{  \textsf{Ct}}

\newcommand{\one}{\ding{202}\xspace}
\newcommand{\two}{\ding{203}\xspace}
\newcommand{\three}{\ding{204}\xspace}
\newcommand{\four}{\ding{205}\xspace}
\newcommand{\five}{\ding{206}\xspace}
\newcommand{\six}{\ding{207}\xspace}
\newcommand{\seven}{\ding{208}\xspace}
\newcommand{\eight}{\ding{209}\xspace}

\newcommand{\height}{h} 

\newcommand{\Obj}{\ensuremath{\mathsf{Obj}}}

\newcommand{\Version}{\ensuremath{\mathsf{Version}}}
\newcommand{\InitVersion}{\ensuremath{\mathsf{InitialVersion}}}

\newcommand{\ObjID}{\mathsf{ObjID}}

\newcommand{\para}[1]{\vskip 1em \noindent \textbf{#1.}}

\newif \ifcomments \commentstrue
\newif \iffull \fulltrue

\fulltrue

\ifcomments
  \newcommand{\note}[2]{\textsf{\color{blue}{[Note(#1): {#2}]}}}
\else
  \newcommand{\note}[2]{}
\fi

\ifdefined\cameraReady
  \def\commentson{0} 
\else
  \def\commentson{1} 
\fi

\ifnum\commentson=1
  \newcommand{\alberto}[1]{{{\color{olive}\textbf{Alberto: }#1}\normalcolor}}
  \newcommand{\lef}[1]{{{\color{pblue}\textbf{Lef: }#1}\normalcolor}}
  \newcommand{\george}[1]{{{\color{orange}\textbf{George:}#1}\normalcolor}}
  \newcommand{\kostas}[1]{{{\color{magenta}\textbf{Kostas:}#1}\normalcolor}}
  \newcommand{\arnab}[1]{{{\color{teal}\textbf{Arnab:}#1}\normalcolor}}

\else
  \newcommand{\alberto}[1]{}
  \newcommand{\lef}[1]{}
  \newcommand{\george}[1]{}
  \newcommand{\kostas}[1]{}
  \newcommand{\arnab}[1]{}
\fi

\definecolor{pblue}{rgb}{0.13,0.13,1}
\definecolor{pgreen}{rgb}{0,0.5,0}
\definecolor{pred}{rgb}{0.9,0,0}
\definecolor{pgrey}{rgb}{0.46,0.45,0.48}

\definecolor{dullred}{Hsb}{0,1,0.4}
\definecolor{dullyellow}{Hsb}{30,1,0.4}
\definecolor{dullgreen}{Hsb}{60,1,0.4}
\definecolor{dullteal}{Hsb}{150,1,0.4}
\definecolor{dullblue}{Hsb}{210,1,0.4}
\definecolor{dullpurple}{Hsb}{270,1,0.4}
\definecolor{dullmagenta}{Hsb}{300,1,0.4}
\definecolor{dullpurplered}{Hsb}{330,1,0.4}

\definecolor{ckeyword}{HTML}{7F0055}
\definecolor{ccomment}{HTML}{3F7F5F}
\definecolor{cnumber}{HTML}{2A0099}

\lstdefinelanguage{Move}{
  keywords={
      abort, acquires, assert, copy, borrow_global, borrow_global_mut, create_account, freeze, fun,
      module, move_to, move_from, public, resource, break, exists, has, key, store, vector
      continue,
      else,
      false,
      if,
      let, loop,
      move, mut,
      return,
      struct,
      true,
      while,
      use,
    },
  ndkeywords={address, id, u8, bool, u64, bytearray, Self},
  showspaces=false,
  showtabs=false,
  breaklines=true,
  showstringspaces=false,
  breakatwhitespace=true,
  lineskip=-0.6pt,
  morecomment=[l]{//}, 
  morecomment=[s]{/*}{*/}, 
  basewidth={0.54em, 0.4em},
  basicstyle= \ttfamily,
  keywordstyle={\color{ckeyword}\ttfamily\bfseries},
  ndkeywordstyle={\color{pblue}\ttfamily\bfseries},
  commentstyle={\color{ccomment}\itshape},
  stringstyle=\color{green},
  moredelim=[il][\textcolor{pgrey}]{$ $},
  moredelim=[is][\textcolor{pgrey}]{\%\%}{\%\%}
}

\definecolor{eclipseStrings}{RGB}{42,0.0,255}
\definecolor{eclipseKeywords}{RGB}{127,0,85}
\colorlet{numb}{magenta!60!black}

\lstdefinelanguage{json}{
  basicstyle=\normalfont\ttfamily,
  commentstyle=\color{eclipseStrings}, 
  stringstyle=\color{eclipseKeywords}, 
  stepnumber=1,
  numbersep=8pt,
  showstringspaces=false,
  breaklines=true,
  string=[s]{"}{"},
  comment=[l]{:\ "},
  morecomment=[l]{:"},
  literate=
    *{0}{{{\color{numb}0}}}{1}
    {1}{{{\color{numb}1}}}{1}
    {2}{{{\color{numb}2}}}{1}
    {3}{{{\color{numb}3}}}{1}
    {4}{{{\color{numb}4}}}{1}
    {5}{{{\color{numb}5}}}{1}
    {6}{{{\color{numb}6}}}{1}
    {7}{{{\color{numb}7}}}{1}
    {8}{{{\color{numb}8}}}{1}
    {9}{{{\color{numb}9}}}{1}
}

\newtheorem{theorem}{Theorem}
\newtheorem{lemma}{Lemma}

\begin{document}

\title[\sysname]{
    \sysname: A Blockchain Combining Broadcast and Consensus
}

\ifdefined\cameraReady
    \author{Sam Blackshear}
    \email{sam@mystenlabs.com}
    \orcid{0000-0002-0024-6655}
    \affiliation{\institution{Mysten Labs}\city{Palo Alto}\country{USA}}

    \author{Andrey Chursin}
    \email{andrey@mystenlabs.com}
    \orcid{0009-0005-2758-5496}
    \affiliation{\institution{Mysten Labs}\city{Palo Alto}\country{USA}}

    \author{George Danezis}
    \email{george@mystenlabs.com}
    \orcid{0000-0002-8923-4860}
    \affiliation{\institution{Mysten Labs, UCL }\city{London}\country{UK}}

    \author{Anastasios Kichidis}
    \email{tasos@mystenlabs.com}
    \orcid{0000-0002-6998-4031}
    \affiliation{\institution{Mysten Labs}\city{London}\country{UK}}

    \author{Lefteris Kokoris-Kogias}
    \email{lefteris@mystenlabs.com}
    \orcid{0000-0002-8827-3382}
    \affiliation{\institution{Mysten Labs, IST Austria}\city{Athens}\country{Greece}}

    \author{Xun Li}
    \email{xun@mystenlabs.com}
    \orcid{0009-0005-0522-8295}
    \affiliation{\institution{Mysten Labs}\city{Palo Alto}\country{USA}}

    \author{Mark Logan}
    \email{mark@mystenlabs.com}
    \orcid{0009-0000-1379-4733}
    \affiliation{\institution{Mysten Labs}\city{Palo Alto}\country{USA}}

    \author{Ashok Menon}
    \email{ashok@mystenlabs.com}
    \orcid{0009-0007-1372-0479}
    \affiliation{\institution{Mysten Labs}\city{London}\country{UK}}

    \author{Todd Nowacki}
    \email{tmn@mystenlabs.com}
    \orcid{0009-0006-8870-8872}
    \affiliation{\institution{Mysten Labs}\city{Palo Alto}\country{USA}}

    \author{Alberto Sonnino}
    \email{alberto@mystenlabs.com}
    \orcid{0000-0001-5337-4741}
    \affiliation{\institution{Mysten Labs, UCL }\city{London}\country{UK}}

    \author{Brandon Williams}
    \email{brandon@mystenlabs.com}
    \orcid{0009-0006-3883-7334}
    \affiliation{\institution{Mysten Labs}\city{Palo Alto}\country{USA}}

    \author{Lu Zhang}
    \email{lu@mystenlabs.com}
    \orcid{0009-0000-4359-3174}
    \affiliation{\institution{Mysten Labs}\city{Palo Alto}\country{USA}}

    \renewcommand{\shortauthors}{Sam Blackshear et al.}
\else
    \author{}
\fi

\begin{abstract}
    \sysname is the first smart-contract platform to sustainably achieve sub-second finality. It achieves this significant decrease by employing consensusless agreement not only for simple payments but for a large variety of transactions. Unlike prior work, \sysname neither compromises expressiveness nor throughput and can run perpetually without restarts. \sysname achieves this by safely integrating consensuless agreement with a high-throughput consensus protocol that is invoked out of the critical finality path but ensures that when a transaction is at risk of inconsistent concurrent accesses, its settlement is delayed until the total ordering is resolved.
    Building such a hybrid architecture is especially delicate during reconfiguration events, where the system needs to preserve the safety of the consensusless path without compromising the long-term liveness of potentially misconfigured clients. We thus develop a novel reconfiguration protocol, the first to provably show the safe and efficient reconfiguration of a consensusless blockchain. \sysname is currently running in production and underpins the Sui smart-contract platform. Combined with the use of Objects instead of accounts it enables the safe execution of smart contracts that expose objects as a first-class resource. In our experiments \sysname achieves latency lower than 0.5 seconds for throughput up to 5,000 certificates per second (150k ops/s with transaction blocks), compared to the state-of-the-art real-world consensus latencies of 3 seconds. Furthermore, it gracefully handles validators crash-recovery and does not suffer visible performance degradation during reconfiguration.
\end{abstract}



\maketitle


\section{Introduction}
Traditional blockchains totally order transactions across replicated miners or validators to mitigate ``double-spending'' attacks, i.e., a user trying to use the same coin in two different transactions. It is well known that total ordering requires consensus. In recent years, however, systems based on consistent~\cite{BaudetDS20} and reliable~\cite{GuerraouiKMPS19} broadcasts have been proposed instead. These rely on objects (e.g., a coin) being controlled by a single authorization path (e.g., a single signer or a multi-sig mechanism), responsible for the liveness of transactions. This concept has been used to design asynchronous, and lightweight alternatives to traditional blockchains for decentralized payments~\cite{astro,BaudetDS20,zef}. We call these systems \emph{consensusless} as they forgo consensus. Yet, so far they have not been used in a production blockchain.

There are multiple reasons for this. First, consensusless protocols typically support a restricted set of operations limited to asset transfers. Second, deploying consesusless protocols in a dynamic environment is challenging as they do not readily support state checkpoints and validator reconfiguration. Supporting these functions is vital for the health of a long-lived production system.
Finally, consensusless protocols are sensitive to client bugs as equivocations lock the assets forever.

Consequently all existing blockchains implement consensus-based protocols that allow for general-purpose smart contracts. Sadly, this flexibility comes at the cost of higher complexity and significantly higher latency even for transactions that operate on unrelated parts of the state.

In this paper, we present \sysname, the first system that combines the consensusless and consensus-based approaches to provide the best of both worlds when processing transactions in a replicated Byzantine setting.
%
\sysname uses a consistent broadcast protocol between validators to ensure the safety of all operations, ensuring lower latency as compared to consensus. It only relies on consensus for the correct execution of complex smart contracts operating on shared-ownership objects, as well as to support network maintenance operations such as defining checkpoints and reconfiguration. It is maintained by a permissionless set of validators that play a similar role to miners in Bitcoin.

\para{Challenges} \sysname requires tackling three key challenges:
%

\noindent Firstly, a high-throughput system such as \sysname requires a checkpointing protocol in order to archive parts of its history and reduce the memory footprint and bootstrap cost of new participants. Checkpointing \sysname however is not as simple as in classic blockchains since it does not have total ordering guarantees for all transactions. Instead, \sysname proposes an after-the-fact checkpointing protocol that eventually generates a canonical sequence of transactions in blocks, without delaying execution and finality.

Secondly, consensusless protocols provide low latency at the cost of usability. A misconfigured client (e.g., underestimating the gas fee or crash-recovering) risks deadlocking its account. We consider this an unacceptable compromise for a production system. We develop \sysname such that client bugs only affect the liveness of equivocated owned objects for a single \emph{epoch}, and provide rigorous proofs to support it. The current epoch length for our production system is 24h, but experiments with an epoch length of 10 minutes show no deterioration in performance.

Finally, the last challenge to solve is the dynamic participation of validators in a permissionless system. The lack of total ordering makes the solution non-trivial as different validators may stop processing transactions at different points compromising the liveness of the system.
Additional challenges stem from the non-starvation needs of misconfigured clients coupled with ensuring that final transactions are never reverted across reconfiguration events. To this end, we design a custom reconfiguration protocol that preserves safety with minimal disruption of the processing pipeline.

\para{Real-world system}
\sysname has been designed for and adopted as the core system behind the Sui blockchain~\cite{sui}, that uses the Move programming language~\cite{move_white} for smart contracts. It launched in May 2023, and has been operating continuously using the protocols described here, with no downtime for a year, as of May 2024.
\bragging

Over the history of \sysname, about 64\% of transactions used the fast path, with peaks reaching 98\% during specific activities such as launches of new applications (e.g., epoch 103). However, in some rare days where traffic is dominated by more traditional DeFi and Oracle applications, the rate of transactions using the fast path drops to 1\%-2\%~\cite{sui-explorer}.
However, on the day (July 2023) the peak throughput of 65 million certificates was observed, the vast majority of these transactions used the fast path, encoding an action of an on-chain game. More recently, a Sui community member\footnote{\url{https://www.youtube.com/watch?v=oWZ9eNUclxA}} showed that 1 million NFTs may be minted on the Sui testnet in about 45 seconds, using the consensussless path. Finally, we note that the number of consensuless transactions may be a poor proxy for their value: transfers and mints are often user facing, and users have low tolerance for delays in payments and transfers (less than a second). While operations relating to DeFi may have a higher natural delay tolerance (a few seconds).

Due to \sysname underpinning an operational infrastructure, the long version of this paper~\cite{sui-lutris} present details that go beyond merely illustrating core components.
%
%
\sysname achieves finality within 0.5 seconds with a committee of 10 validators processing up to 5,000 cert/s or with a committee of 100 validators processing about 4,000 cert/s (150k ops/s with transaction blocks). \sysname can provably withstand up to $1/3$ of crash validators without meaningful performance degradation and can seamlessly reconfigure despite the complexities of supporting a consesusless transactions.

We evaluated \sysname against Bullshark~\cite{bullshark}, a state-of-the-art consensus protocol, running within the same blockchain and show that \sysname achieves finality up to 15x earlier. The \sysname execution engine leverages the explicit transaction metadata about objects accesses to statically schedule transaction execution in parallel on all cores.
Consequently, \sysname simplifies dependency-graph extraction during execution by seamlessly tracking accessed objects, mitigating congestion across different shared-objects.
Parallel execution is a desired feature that both EVM chains~\cite{saraph2019empirical} and pure consensus chains~\cite{gelashvili2023block} are currently trying to adopt. However details of the parallel execution are out of scope.

\para{Contributions} We make the following contributions:
\begin{itemize}[leftmargin=*]
      \item We present \sysname, the first smart-contract system that forgoes consensus for single-writer operations and only relies on consensus for settling multi-writer operations, combining the two modes securely and efficiently.

      \item As part of \sysname we show how to use a consensus engine to efficiently checkpoint a consensusless blockchain without forfeiting the latency benefits of running consensusless transactions. Our checkpointing mechanism puts transactions into a sequence after execution and finality, reducing the need for agreement and therefore latency.

      \item We show how to perform reconfiguration safely and with minimal downtime. Unlike prior consensusless blockchains our reconfiguration mechanism allows for forgiving equivocation so that careless users can regain access to their assets.




      \item We provide a production-grade implementation of \sysname and evaluate it, on a real geo-distributed set of validators, under varying transaction loads, and crash faults.
\end{itemize}
%

\section{Background \& Models} \label{sec:overview}

\sysname uses a novel approach to processing blockchain transactions which ensures low latency 
by forgoing the need for consensus from the critical latency path. 
Yet, skipping consensus provides finality (knowing that a transaction will execute) and settlement (knowing the exact execution result)
only for single-owner assets (assets that are immutable or owned directly or indirectly by a single address, see below).

Despite single-owner transactions constituting most of the load on our mainnet, these types of assets are not expressive enough to implement all types of smart contracts since some transactions must process assets belonging to different parties. For the rest of the transactions, we can only get finality but need to postpone settlement until potential conflicts are resolved. For this reason, we couple \sysname with a consensus protocol. This hybrid architecture is both a curse and a blessing: two different execution paths create the threat of inconsistencies and safety concerns; but, the consensus component enables us to implement checkpointing and reconfiguration, which consensusless systems lack.

We define the problem \sysname addresses, the threat model, and the security properties maintained.

\subsection{Threat Model}

We assume a message-passing system with a set of $n$ validators per epoch and
a computationally bound adversary that controls the network and can statically
corrupt up to $f < n/3$ validators within any epoch.
We say that validators corrupted by the adversary are \emph{Byzantine} or \emph{faulty} and
the rest are \emph{honest} or \emph{correct}.
To capture real-world networks we assume asynchronous \emph{eventually
    reliable} communication links among honest validators.
That is, there is no bound on message delays and there is a finite but unknown number of messages that can be lost.
Informally, \sysname exposes to all participants a \emph{key-value} object store abstraction that can be used to read and write objects.


\noindent \textbf{Consensus Protocol.}
\sysname uses a consensus protocol as a black box that takes some valid inputs and outputs a total ordering. It makes no additional synchrony assumptions and thus inherits the synchrony assumptions of the underlying consensus protocol. In our implementation, we specifically use the Bullshark protocol~\cite{bullshark}. It is secure in the partially synchronous network model~\cite{dwork1988consensus}, which stipulates that after some unknown global stabilization time all messages are delivered within a bounded delay.
It could be configured to run the Tusk protocol~\cite{narwhal}, making \sysname asynchronous.

\subsection{Core Properties} \label{sec:core-properties}
\vspace{-0.2cm}

\sysname achieves the standard security properties of blockchain systems relating to validity, safety, and liveness:
\begin{itemize}[leftmargin=*]
    \item \textbf{Validity}: State transitions at correct validators are in accordance with the authorization rules relating to objects, as well as the VM logic constraining valid state transitions on objects of defined types. This property is unconditional with respect to the number of correct validators in the network.

    \item \textbf{Safety}: If two transactions $t$ and $t'$ are executed on correct validators, in the same or different epochs, and take the same inputs, then $t=t'$. This property also holds in asynchrony subject to a maximum threshold of Byzantine nodes.

    \item \textbf{Liveness}: All valid transactions sent by correct clients are eventually processed until final (and their effects persist across epoch boundaries).  All objects that have not been used as inputs to a committed transaction are eventually available to be used by a correct client as part of a valid transaction. This property holds under partial synchrony, due to our use of Bullshark~\cite{bullshark} consensus (but would hold in asynchrony when using Tusk~\cite{narwhal}).
\end{itemize}

Liveness encompasses censorship resistance. It also only holds for a correct client that does not equivocate by sending conflicting transactions for the same owned object version.


\subsection{Transaction Finality}\label{sec:bis-finality}
The BIS~\cite{BISfinality} defines finality as the property of being ``irrevocable and unconditional''. We distinguish between transaction finality, after which a transaction processing is final,
and settlement, after which the effects of a transaction are final and may be used by subsequent transactions. In \sysname, unlike other blockchains, both finality and settlement occur before the creation of checkpoints.

In all cases, \emph{a transaction becomes final when $2f+1$ validators accepts for processing a transaction counter-signed by $2f+1$ validators (i.e., a certificate), even before such a certificate is sequenced by consensus or executed}. After this, no conflicting transaction can occur, the transaction may not be revoked, and is eventually executed and persists across epochs. However, the result of execution is only known a-priori for owned object transactions, and for shared object transactions it is known only after consensus (see \Cref{sec:design-description}). Transaction finality is achieved within 2 network round trips.

Settlement occurs upon execution on $2f+1$ validators, when an \emph{effects certificate} could be formed (see \Cref{sec:design-description}).
For owned object
transactions execution occurs without the delay of consensus; for shared object transactions it happens just after the certificate
has been sequenced by consensus. In both cases, settlement is not delayed by the process of committing the transaction within a checkpoint, and thus, has lower latency
than checkpoint creation.

\section{The \sysname Architecture} \label{sec:core-protocol}
We present the \sysname protocol, which combines broadcast and consensus to provide a scalable and secure blockchain. We defer detailed algorithms, explanations, and data structures to extended version paper~\cite{sui-lutris}.

\subsection{Object-Centric Design} \label{sec:system-model}
Current blockchains typically use two types of programming models for their core resource. Unspent Transaction Outputs (UTXO)~\cite{nakamoto2008bitcoin} or accounts~\cite{zhang2022ethereum}.
Neither of the two serves well for \sysname's goals as UTXO provides a bad user experience where the wallets need to keep track of all the fresh UTXO entries owned by the user whereas Accounts limit parallelism as a user's account cannot be concurrently credited, which would make consensusless path transactions significantly less usable.

\para{\sysname objects}
To this end, we introduce a middle ground, which we simply call Objects. Objects are long-lived, like accounts, but can be manipulated atomically through versioning, allowing a significant amount of parallelism. More specifically,  \sysname validators replicate the state represented as a set of objects. Each object has a type
that defines operations that are valid state transitions for the type. Each object may be read-only, owned, or shared.
\begin{itemize}[leftmargin=*]
    \item \emph{Read-only objects} cannot be mutated within an epoch and can be used in transactions concurrently and by all users.
    \item \emph{Owned objects} have an owner field. The owner can be set to an address representing a public key.
          In that case, a transaction is authorized to use the object and mutate it if it is signed by that address. A transaction is signed by a single address and, therefore, can use objects owned by that address. However, a single transaction cannot use objects owned by multiple addresses. The owner of an object (called a child object) can also be another object (called the parent object).  In that case, the child object may only be used if the root object (the first one in a tree of possibly many parents) is part of the transaction and the transaction is authorized to use the parent. Contracts use this facility to create dynamic collections and other complex data structures.
    \item \emph{Shared objects} are mutable and do not have a specific owner. They can instead be included in transactions by anyone, and they perform their own authorization logic as part of the smart contract. Such objects, by virtue of having to support multiple writers while ensuring safety and liveness, require a consensus protocol to be executed safely.
\end{itemize}

The extended version of the paper~\cite{sui-lutris} provides additional details on the object model and the operations that can be performed on objects, including objects wrapping/unwrapping, parent-child relationships, and details on safe objects deletion.

\para{Object versions}
Both owned and shared objects are associated with a version number. The tuple $(\objectid, \objectversion)$ is called an $\objectkey$ and takes a single value, thus it can be seen as the equivalent of a Bitcoin UTXO~\cite{nakamoto2008bitcoin} that should not be equivocated. This explicit ownership of Objects allows \sysname increased programming simplicity. More concretely:

\begin{itemize}
    \item  It improves security as, by construction, an attacker cannot steal or even send a transaction that accesses an owned object--only the owner address can. This is in contrast to the other programming models where all smart contract data is shared by default and only through (often delicate or incorrect) access control is an attacker limited.
    \item It enables native composability. An object can both wrap and own other objects, providing a natural mechanism for creating rich ownership hierarchies. For example, a game character may have an inventory containing swords, shields, and potions, and the sword might own badges  recording its historical achievements.
    \item  It defines a common vocabulary for assets throughout the client stack. Assets are the building blocks of decentralized applications, and they need to be surfaced to users. Yet on existing blockchains, there is no explicit representation of assets. Instead, application-specific view functions read unstructured, untyped data from accounts and attempt to convert it into types the client can understand. On \sysname, objects are a typed, structured representation of assets and every object has a unique cryptographic identifier. Objects serve as a unifying abstraction: programmers create objects, APIs read them, smart contracts mutate them, wallets and apps display them.
    \item It allows an easy way to detect when two transactions do not have data dependencies since they do not share owners. This allows maximal use of \sysname's fast path as well as easily parallelizable execution of \sysname's workloads in general, which can now leverage the multi-core architecture of modern CPUs.
    \item It allows to statically check whether a transaction can forgo consensus: transactions forgo consensus if they only access owned objects (see \Cref{sec:design-description}).
\end{itemize}

Object-level tracking of state draws a balance between the account based model that limits parallelism to the unit of authentication (i.e., the account), and UTXO, which allows parallelism at the cost of mutating identifiers, making state difficult to track across time. The object-centric design allows to detect read-write/write-write conflicts statically for shared-objects and resolve them after sequencing by assigning to transactions the shared-object version on which they should operate, which is extremely performant. To our knowledge, no previous system did this.

However, the object-centric model also introduces some complexities for the users: they need to track multiple objects they own instead of a single account, which makes managing coin balances harder, while making other operations dealing with objects easier. This is mitigated in Sui via wallets managing the complexity and showing a unified balance, indexers reporting aggregate balances per account, all transactions being able to use gas from multiple objects and smashing them into a single object, and programable transactions blocks that can take multiple object inputs, aggregate them, and use them as a single object.

\para{Transactions} A transaction is a signed command specifying the input objects (read-only, owned or shared), a version number for each owned object, an entry function into a 
smart contract, and a set of parameters. If valid, it consumes the mutable input objects and constructs a set of output objects at a fresh version -- which can be the same objects at a later version or new objects.

\subsection{Example Owned and Shared Object Uses}
As a simple example of the use of owned objects we show in \Cref{fig:owned-contract} a game character representing a warrior. The public functions `new\_warrior' and `new\_sword' respectively create a new warrior object and a sword object to equip to the warrior. Both these objects are owned by the sender of the transaction. Calls to `equip` and `unequip` mutate the warrior to respectively equip and unequip them with a sword.
We are open sourcing a production-ready example of game character implemented exclusively with owned objects\footnote{
    \ifdefined\cameraReady
        \url{https://github.com/MystenLabs/sui/blob/main/examples/move/hero/sources/example.move}
    \else
        Code available but link omitted for review.
    \fi
}.

\begin{figure}[t]
    \begin{lstlisting}[language=Move, basicstyle=\tiny\ttfamily, columns=fullflexible]
/// Example of a game character with a basic attribute.
module example::warrior {
    /// A sword to equip to our warrior
    struct Sword has key, store {
        id: UID,
        strength: u8,
    }
    /// A owned object representing a warrior
    struct Warrior has key, store {
        id: UID,
        sword: Option<Sword>,
    }

    /// Create a new sword
    public fun new_sword(strength: u8, ctx: &mut TxContext) {
        let sword = Sword { id: object::new(ctx), strength };
        transfer::transfer(sword, tx_context::sender(ctx));
    }
    /// Create a new warrior
    public fun new_warrior(ctx: &mut TxContext) {
        let warrior = Warrior { id: object::new(ctx), sword: option::none() };
        transfer::transfer(warrior, tx_context::sender(ctx));
    }

    /// Equip the warrior with a sword
    public entry fun equip(warrior: &mut Warrior, sword: Sword) {
        assert!(option::is_none(&warrior.sword), EAlreadyEquipped);
        option::fill(&mut warrior.sword, sword);
    }
    /// Remove the sword from the warrior
    public entry fun unequip(warrior: &mut Warrior): Sword {
        assert!(option::is_some(&warrior.sword), ENotEquipped);
        option::extract(&mut warrior.sword)
    }
}
\end{lstlisting}
    \caption{Owned object contract for inventory management.}\label{fig:owned-contract}
\end{figure}

As an example of the use of shared objects we show in \Cref{fig:shared} an atomic swap contract. The public `create' function wraps an escrowed item into a \emph{shared} escrowed Object, to be swapped for another object with a specific ID. One of two things can then happen: the counter party can call `exchange' with the correct object to swap executes the atomic swap; or the creator of the escrow object can call `cancel' and recovers the escrowed object. Since two distinct mutually distrustful actors may access this state it has to be represented as a shared object.

\begin{figure}[t]
    \begin{lstlisting}[language=Move, basicstyle=\tiny\ttfamily, columns=fullflexible]
/// An escrow for atomic swap of objects without a trusted third party
module defi::shared_escrow {
    /// An object held in escrow
    struct EscrowedObj has key, store {
        id: UID,
        /// owner of the escrowed object
        creator: address,
        /// intended recipient of the escrowed object
        recipient: address,
        /// ID of the object `creator` wants in exchange
        exchange_for: ID,
        /// the escrowed object
        escrowed: Option<T>,
    }

    /// Create an escrow for exchanging goods with counterparty
    public fun create(recipient: address, exchange_for: ID, escrowed_item: T, 
        ctx: &mut TxContext
    ) {
        let creator = tx_context::sender(ctx);
        let id = object::new(ctx);
        let escrowed = option::some(escrowed_item);
        transfer::public_share_object(
            EscrowedObj { id, creator, recipient, exchange_for, escrowed } );
    }

    /// The `recipient` of the escrow can exchange `obj` with the escrowed item
    public entry fun exchange( obj: ExchangeForT, escrow: &mut EscrowedObj,
        ctx: &TxContext
    ) {
        assert!(option::is_some(&escrow.escrowed), EAlreadyExchangedOrCancelled);
        let escrowed_item = option::extract(&mut escrow.escrowed);
        assert!(&tx_context::sender(ctx) == &escrow.recipient, EWrongRecipient);
        assert!(object::borrow_id(&obj) == &escrow.exchange_for, EWrongExchangeObject);
        // everything matches. do the swap!
        transfer::public_transfer(escrowed_item, escrow.recipient);
        transfer::public_transfer(obj, escrow.creator);
    }

    /// The `creator` can cancel the escrow and get back the escrowed item
    public entry fun cancel(escrow: &mut EscrowedObj,
        ctx: &TxContext
    ) {
        assert!(&tx_context::sender(ctx) == &escrow.creator, EWrongOwner);
        assert!(option::is_some(&escrow.escrowed), EAlreadyExchangedOrCancelled);
        transfer::public_transfer(option::extract(&mut escrow.escrowed), escrow.creator);
    }
}
\end{lstlisting}
    \caption{Shared object contract for an atomic swap.}\label{fig:shared}
\end{figure}

\subsection{Core Protocol Description} \label{sec:design-description}
\Cref{fig:sui-overview} illustrates the high-level interactions between a client and \sysname validators to commit a transaction.

\begin{figure*}[t]
    \centering
    \includegraphics[width=0.8\textwidth]{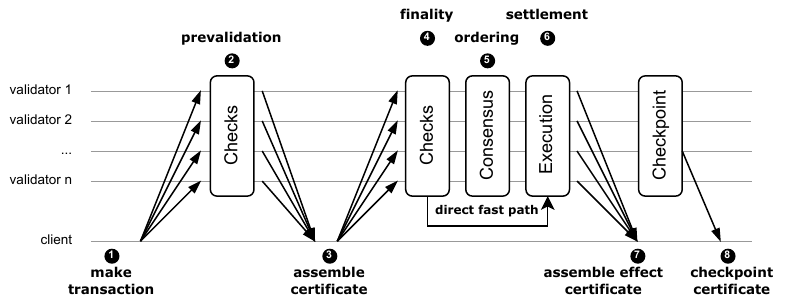}
    \caption{Overview and transaction life cycle. The Byzantine agreement protocol is only executed for transactions containing shared objects and is not necessary for transactions involving only single-owner objects.}
    \label{fig:sui-overview}
\end{figure*}

\para{Step~\one of \Cref{fig:sui-overview}: Dispatch transactions}
A user with a private key creates and signs a \emph{user transaction} to either mutate objects they own (case 1) or a mix of objects they own (at least one is gas) and shared objects (case 2). This is the only time user signature keys are needed; the remaining process may be performed by the user or a gateway on behalf of the user.

\para{Step~\two of \Cref{fig:sui-overview}: Prevalidate transactions}
The first challenge \sysname needs to overcome is transaction spamming, occurring when malicious clients submits transactions referencing non-existing or already-spent objects. These attack are currently inherent on blockchains that execute post-ordering such as DAGs~\cite{spiegelman2022bullshark,narwhal,stathakopoulou2022state,celestia} as the decoupling that increases their performance also means that they simply orders bytes which could be duplicate transactions.
In \sysname we circumvent this issue without requiring full execution by adding a prevalidation step before execution. This prevalidation step makes sure that only transactions with gas are forwarded to the resource-heavy consensus protocol.

Specifically, the signed transaction is sent to each \sysname validator in order to check it for validity. If it is valid then the validator locally locks all input owned objects using their $\objectkey$, signs it, and returns the \emph{signed transaction} to the client. 
It is critical to see that since an $\objectkey$ is locked for a specific transaction, no double-spend can happen and only one transaction per object can pass this check, preventing spams. \sysname leverages this mechanism to force clients to lock gas objects, thus limiting the ability to spam the network.

\para{Step~\three of \Cref{fig:sui-overview}: Certificate aggregation}
A second benefit of our spam prevention approach is the low cost it imposes on validators. Once they sign the transaction it is the job of the client (or any external gateway entity) to
collect the responses from a quorum ($2f+1$) of validators to form a \emph{transaction certificate}. As a result, unlike consensus-based blockchains, in \sysname, the validators do not need to gossip signatures or aggregate certificates, which is now the responsibility of the client/gateway.

A second challenge that \sysname had to address is the fact that enabling two parallel transaction processing paths has the risk of breaking the safety of the system if two transactions try to access the same object on different paths. A naive solution to this would be to shard the state into two such that owned objects only interact with owned objects and shared objects only interact with shared objects.
However, this would make the system effectively two fragmented systems that are simply running on the same hardware and cannot natively interact with one another.
Instead in \sysname we exploit the fact that both communication paths run on the same hardware to synchronize using shared memory. Specifically, every validator has locally a lock table named $\ownedlockdb$, which assigns objects to transactions during the prevalidation step. Together with the fact that all transactions need to go through the prevalidation step, we can ensure that there is no concurrent access on an $\objectkey$ (more details on the extended version of the paper~\cite{sui-lutris}).

\para{Step~\four of \Cref{fig:sui-overview}: Transaction finality}
Once the prevalidation step is successful and the client has assembled the certificate, it sends it back to all validators, that respond once they check its validity (a quorum of responses at this point ensures
\emph{transaction finality}, see \Cref{sec:bis-finality}). If the transaction involves exclusively read-only and owned objects the transaction certificate can be immediately executed and settled through a parallel execution (shown as the \emph{direct fast path} in the diagram), thus skipping step (\five) of \Cref{fig:sui-overview}. As demonstrated in \Cref{sec:evaluation}, this path typically leads to settling transactions in sub-second latencies almost 4x less than running them on consensus.

\para{Step~\five of \Cref{fig:sui-overview}: Total ordering}
All certificates are forwarded to a \emph{Byzantine agreement protocol} operated by the \sysname validators. Consensus then outputs a total order of certificates, and validators check the certificate and schedule the parallel execution of shared object transactions through a lock table (see the extended version~\cite{sui-lutris} for details, the table is named $\sharedlockdb$).

\para{Step~\six of \Cref{fig:sui-overview}: Transaction settlement}
Validators leverage the tables $\ownedlockdb$ and $\sharedlockdb$ to execute transactions in parallel. The execution result is a summary of how the transaction affects the state and is used to construct a \emph{signed effects} response. 

\para{Step~\seven of \Cref{fig:sui-overview}: Effect certificate}
Once a quorum of validators has executed the certificate, its effects are final (in the sense of \emph{settlement}, see \Cref{sec:bis-finality}). Clients can collect a quorum of validator responses, create an \emph{effects certificate}, and use it as proof of the finality of the transaction effects.

\subsection{Cost of Transactions}
The final resulting protocol, \sysname, manages to ensure that execution for transactions involving read-only and owned objects requires only reliable broadcast and a single certificate. This has a minimal $O(n)$ communication and computation cost on the critical path and requires \emph{no validator to validator communication}. Smart contract developers can therefore design their types and their operations to optimize transfers and other operations on objects of a single user to reduce the cost of their transactions.
At the same time, \sysname also provides the flexibility of using shared objects, through the Byzantine agreement path, and enables developers to implement logic that needs to be accessed by multiple users.

\subsection{Mitigating User Errors}
Equivocation of an $\objectkey$ is considered malicious in prior work~\cite{BaudetDS20, astro, zef} and permanently locks the object forever. However, equivocation is often the result of a misconfigured wallet due to poor synchronization between gateways or poor gas predictions, leading to re-submissions. As we discuss in \Cref{sec:reconfiguration}, \sysname only loses liveness for equivocated owned objects until the end of the epoch. In the new epoch, the user can try again to access the same $\objectkey$ with a fresh, hopefully correct, transaction. This is in contrast with all other consensusless blockchain that punishes all kinds of equivocation with a deadlock of the asset forever.

\section{Long-Term Stability}
Unlike previous consensus-less systems, \sysname is specifically engineered for sustained production use over extended periods. As of Jun 13, 2024 it has operated continuously with no downtime for 13 months. Achieving this requires implementing protocols that allow for the seamless removal and addition of validators -- a feature not present in earlier consensusless designs that assume infinite memory and static membership~\cite{BaudetDS20,zef, astro}.
In order to ensure long-term stability, we detail the necessary protocols, emphasizing the creation of checkpoints and the facilitation of reconfiguration.

To optimize for minimal latency, \sysname executes prior to block creation. While this enhances speed, it introduces complexity for external entities seeking proofs of transaction execution, conducting comprehensive audits, or systematically replicating the chain's state. To address these challenges, we introduce checkpoints.
Moreover, to handle client-side errors effectively, there is a need to safely unlock objects erroneously locked due to client double-spending bugs. The evolution of the validator set and voting power over time is essential to support permissionless delegated proof of stake. These functionalities are embedded in the epoch-change and reconfiguration process. This process persists all final transactions and their effects across epochs, enabling clients to unlock objects involved in partially locked transactions securely. It provides a robust mechanism for validators to enter or exit the system without compromising its liveness and performance. The period between reconfigurations represents a rare window of perfect consistency among all validators, facilitating software upgrades and the distribution of global incentives and rewards.

\subsection{Checkpoints} \label{sec:checkpoint}\label{sec:checkpoint-protocol}
%
\sysname validators emit a sequence of certified \emph{checkpoints} each
containing an agreed-upon sequence of transactions, the authorization path of the transactions, and a commitment to their effects. These form a hash-chain, which is the closest \sysname has to
the blocks of a traditional blockchain. Checkpoints are used for multiple purposes: they are gossiped from validators to full nodes to update them about the state of the chain; they are used by validators to perform synchronization in case they fall behind in execution; as well as to bring new
validators up to speed with the state at the start of each epoch. Checkpoints, packaged along with the transactions and the effects structures
they contain are also the canonical historical record of execution used for audit.

\para{Step~\eight of \Cref{fig:sui-overview}: Checkpoint creation} Upon receiving a valid certificate a correct validator records it and commits to including it in a checkpoint before the end of the epoch.
A validator schedules all certificates for sequencing using the consensus engine.
For owned transactions this does not block execution, which can happen before the transaction is sequenced; in the case of shared object transactions
execution resumes once the transaction certificate has been sequenced. 

Periodically, using a deterministic rule, validators pick a consensus commit to use as a checkpoint (our current implementation checkpoints each and every commit separately).
The new checkpoint contains all transactions present in the commits
since the last checkpoint and any additional transaction required for the execution to be causally complete. If a transaction's certificate exists more than once, the first occurrence is taken as the canonical one and the accompanied user signature is included in the checkpoint for audit purposes.
Note that due to asynchrony and
failures, the certificates sequenced may not be in causal order, or may contain `gaps', i.e., missing transaction
dependencies. Checkpoint creation waits for all transactions necessary for the checkpoint to be causally complete to be within a commit, sorts them in canonical topological order, and includes them in the checkpoint.


\subsection{Committee Reconfiguration} \label{sec:reconfiguration}

Reconfiguration occurs between epochs when the current committee is replaced by a new committee. Other changes requiring global coordination, such as software updates, and parameter updates also take effect between epochs. Immutable objects, such as system parameters or software packages, may be mutated in that period.

The goal of the reconfiguration protocol is to preserve safety of transactions between epochs, and unlock equivocated objects. To this end, we require that if a transaction $\transaction$ was committed during epoch $e$ or before, no conflicting transaction can be committed after epoch $e$. This is trivial to ensure when running only a consensus protocol since a reconfiguration event logged on-chain clearly separates transactions committed in epoch $e$ from transactions committed in epoch $e+1$. However, in \sysname, solutions are not as straightforward. Specifically, \sysname requires a final transaction at epoch $e$ with effects reflected in all subsequent epochs.

\para{Challenges and security intuition}
The main challenge for \sysname reconfiguration at each validator is the race between committing transactions and constructing checkpoints, that are running  asynchronously to each other. If a checkpoint snapshots the end-state of epoch $e$ at time $T$ and is only committed at time $T+1$, we cannot set that checkpoint as the initial state of epoch $e+1$. If we did, all transactions happening during the last timestamp are at risk of being unsafe when validators drop their $\ownedlockdb$ to allow for liveness recovery of equivocated transactions. This is not an issue for the consensus path since we can define as end state the checkpoint at time $T$ plus all transactions ordered before it is committed at $T+1$. That end state is well-defined by the total ordering property of consensus. Unfortunately, it is hard to establish a set of committed transactions on the consensus-less path.



Remember that the consensusless path works in two phases, first a transaction locks the single-owner object and produces a certificate. Then this certificate is sent as proof to the validators who reply with a signed effects certificate (execution). The safety risk during a reconfiguration is that a transaction is executed during the transition phase without the checkpoint recording it. However, \sysname splits reconfiguration into multiple steps instead of doing it atomically. As part of the last step, we pause the consensus-less path. This allows us to show that if a transaction had executed before the reconfiguration message
then there is no safety risk because we can guarantee that at least one honest party will persist the execution in the state. We enforce this by introducing an \textit{End-of-Epoch} message that the committee members of the current epoch send when they have seen all the certificates they processed in the consensus output sequence. The new committee takes over completely only after $2f+1$ such End-of-Epoch messages are ordered. As a result, \sysname manages to preserve safety without needing to block for long periods. The full proofs can be found in \Cref{sec:security}.

\para{The \sysname reconfiguration protocol} \label{sec:reconfiguration-protocol}
The \sysname reconfiguration logic is coded as the smart contract can be found in the extended version of the paper~\cite{sui-lutris}.
Once a quorum of stake for epoch $e$ votes to end the epoch, authorities exchange information to commit to a checkpoint, determine the next committee, and change the epoch.
More specifically,  reconfiguration happens in four steps, (i) \emph{stake recalculation}, (ii) \emph{ready new committee}, (iii) \emph{End-of-Epoch}, and (iv) \emph{Handover}. Each step is handled by the interaction of the new and old committee members with the smart contract functions. This is key for the correctness of the protocol.
The smart contract is parametrized with  the total number of checkpoints before initiating the epoch change protocol $\textsf{S}$ and the minimum amount of stake $\textsf{T}$ required to become a validator.

\para{Step 1: Registration of new validators}
Candidate validators in the next epoch submit a transaction to the reconfiguration contract calling the \textsc{Register} function. This function establishes the static stake distribution for the next epoch.
The smart contract accepts registrations until the $\texttt{S}^{\text{th}}$ checkpoint of the epoch is created.

\para{Step 2: Ready new committee}
Before taking over committee operations, future validators run a full node
to download the required state to become a validator and make sure their local state is consistent. 
Once a validator for the new epoch is ready to start validating, they call the \textsc{Ready} function to signal they have successfully synchronized the required state. This function can only be called towards the end of the epoch, after the creation of $\texttt{S}^{\text{th}}$ checkpoints. 
The cutoff period for the epoch is when a quorum of new validators is ready. At this point, the old committee stops signing new transactions or locking objects. 

\para{Step 3: End of epoch}
When the new committee is ready the old committee only runs consensus and their only job is to make sure all the transactions they executed are sequenced by the consensus engines (so that they are part of the next checkpoint). To this end, they stop accepting certificates from clients and instead make sure that all the certificates they have processed are sequenced by the consensus engine.
They then call the \textsc{End-of-Epoch} function.
This means that any transactions submitted by clients for this epoch are discarded with an end-of-epoch message and need to be resent with an updated epoch number.

\para{Step 4: Handover}
After $2f+1$ old validators call the \textsc{End-of-Epoch} function, the system enters the handover phase.
After an extra checkpoint 
anyone can call the function \textsc{Handover} that effectively terminates the epoch.  At this stage, the state of the smart contract is reset and the old validators can shut down. 
If a validator also participates in the new epoch, it must perform the following operations before entering the next epoch: (i) drop the temporary lock stores, and (ii) roll back the execution of any transaction that did not appear in any checkpoint so far. As discussed earlier, this is safe as these transactions were not final.

This protocol design decouples all the essential steps needed for a secure reconfiguration. This decoupling allows for the necessary logic of defining the new committee, providing the new committee sufficient time to bootstrap from the actual handover, and preserving the safety of consensusless transactions across epochs. Service interruption is thus minimized and unobservable (see \Cref{sec:evaluation}).

\section{Security Proofs of \sysname} \label{sec:security}
\sysname satisfies validity, safety, and liveness as informally described in \Cref{sec:core-properties}.
These properties hold against any polynomial-time constrained adversary under the following assumptions:
\begin{itemize}[leftmargin=*]
    \item \textbf{BFT:} Every quorum of $3f+1$ validators contains at most $f$ Byzantine validators. Correct validators of previous epochs never leak their signing keys. 
    \item \textbf{Network:} The network is partially synchronous~\cite{dwork1988consensus}. \sysname operates in the partial synchrony model only due to our use of Bullshark~\cite{bullshark} as consensus protocol (it would operate in asynchrony if we used Tusk~\cite{narwhal} instead).
\end{itemize}
%

Parts of the proofs also refer to the exact algorithms provided in the extended version of the paper~\cite{sui-lutris} for accuracy, but the reader can omit these references.

\subsection{Safety Analysis} \label{sec:safety}
We prove the safety of the \sysname protocol. That is, we show that at the start of every epoch, all correct validators have the same state.

\para{Proof intuition}
We provide an intuition of the safety of the core protocol outlined in \Cref{sec:core-protocol}. \sysname preserves two fundamental properties: (i) prevention of conflicting transactions' execution, and (ii) validators reach the same states, even when transactions on the fast path are executed in different orders.

To ensure the first property, we leverage the well-established quorum intersection argument. No honest validator would permit the certification of two conflicting transactions. Given that all transactions require certification at step (\two) of \Cref{fig:sui-overview}, and considering the existence of $\ownedlockdb$ -- a shared memory communication channel between the fast path and the consensus path on every validator -- we can confidently assert that there will be no conflicts.

Addressing the second property involves demonstrating that \sysname, as the pioneering blockchain deviating from a serial schedule, upholds the state machine replication guarantees crucial for our objectives. This assurance stems from two key factors: (i) for every owned object transaction, the $\ownedlockdb$ table, coupled with monotonically increasing version numbers of each object, ensures linearizability; and (ii) it is evident that despite the non-serial execution schedule, it remains serializable. An easily identifiable equivalent schedule is provided to external entities (e.g., full nodes), to execute transactions based on checkpoints rather than upon certificate reception. Since this schedule represents a total ordering, it inherently follows a serial structure. Furthermore, as each validator executes transactions in a non-conflicting manner aligned with this serial schedule, all schedules maintain serializability, rendering them as secure as any other serial blockchain.

We prove safety using three main ingredients: (i) correct validators execute the same set of transactions, (ii) correct validators execute those transactions in the same order whenever (partial) order matters, and (iii) the execution of those transactions causes the same state transition across all correct validators.

\para{Execution set}
We start by showing that all correct validators eventually execute the same set of transactions.

\begin{lemma}[Owned-Object Execution Set] \label{th:owned-object-execution-set}
    No correct validators have executed a different set of owned object transactions by the end of epoch $e$.
\end{lemma}
\begin{proof}
    Correct validators assemble checkpoints by observing the sequence of transaction certificates committed in consensus (\Cref{sec:reconfiguration-protocol}). By the agreement properties of the consensus protocol, all correct validators obtain the same sequence of certificates.
    They then use those ordered certificates to construct the same checkpoints since creating checkpoints is a deterministic process given the certificate sequence input.
    Correct validators execute all transactions within the checkpoints they assemble.
    Additionally, at the end of every epoch validators revert the execution of any owned-object transaction not included in a checkpoint (see \Cref{sec:checkpoint-protocol}), hence only owned-object transactions included in a checkpoint persist.
    Reverting these transactions is safe since final transactions are eventually all included in a checkpoint within the epoch (see \Cref{th:client-safety}).
\end{proof}

\begin{lemma}[Shared-Object Execution Set] \label{th:shared-object-execution-set}
    No correct validators have executed a different set of shared object transactions by the end of epoch $e$.
\end{lemma}
\begin{proof}
    Correct validators execute all shared object transactions sequenced by the consensus protocol before the last checkpoint of epoch $e$ (see \Cref{sec:reconfiguration}). By the agreement property of the consensus protocol, all correct validators obtain the same sequence and thus execute the same set of shared-object transactions.
\end{proof}

\para{Execution order}
We now show that correct authorities execute conflicting transactions in the same order.
\begin{lemma}[BCB Consistency] \label{th:bcb-consistency}
    No two conflicting transactions, namely transactions sharing the same owned inputs objects, object version, and epoch, are certified.
\end{lemma}
\begin{proof}
    The proof of this lemma directly follows from the consistency property of Byzantine consistent broadcast (BCB)~\cite{cachin2011introduction} over the label $(e, \ObjID, \Version)$.
    Let's assume two conflicting transactions $\transaction_A$ and $\transaction_B$ taking as input the same owned object $\Obj$ with version $\Version$ are certified during the same epoch $e$.
    Then $f+1$ correct validators produced $\txsign_A$ and $f+1$ correct validators did the same and produced $\txsign_B$.
    A correct validator cannot successfully process both (conflicting) $\transaction_A$ and $\transaction_B$. As a result, a set of $f+1$ correct validators produced $\txsign_A$ but not $\txsign_B$, and a distinct set of $f+1$ correct validators produced $\txsign_B$ but not $\txsign_A$.
    Hence there should be $f+1+f+1=2f+2$ correct validators in addition to the $f$ byzantine. However, $N=3f+1 < 3f+2$ hence a contradiction.
\end{proof}


\begin{lemma}[Shared-Locks Consistency] \label{th:shared-locks-consistency}
    The shared lock stores 
    of correct validators are never conflicting; that is, the shared lock stores of correct validators are either identical or a subset of each other.
\end{lemma}
\begin{proof}
    Let's assume two correct validators update their shared lock stores 
    by assigning different version numbers to a shared object $\Obj$ of a certificate $\cert$.
    Correct validators only assign a version number to $\Obj$ after sequencing $\cert$ through consensus.
    Once consensus outputs the call, a local function named \textsc{AssignSharedLocks} which is deterministic. Consequently, the version number assigned to $\Obj$ depends only on the consensus output sequence.
    %
    %
    As a result, two correct validators assign a different version to $\Obj$ only if they receive a different consensus output sequence. However, by the agreement property of the consensus, all correct validators received the same output sequence, hence a contradiction.
\end{proof}

\begin{lemma}[Owned Objects Sequential Execution] \label{th:sequencial-exec-owned}
    If two certificates $\cert$ and $\cert'$ both take as input the same owned object $\Obj$ (and no shared objects), all correct validators execute $\cert$ and $\cert'$ in the same order.
\end{lemma}
\begin{proof}
    Let's assume two correct validators $v$ and $v'$ execute in different orders $\cert$ and $\cert'$ taking the same input object $\Obj$. That is, $v$ executes $\cert$ then $\cert'$, and $v'$ executes $\cert'$ then $\cert$.
    We argue this lemma by contradiction of \Cref{th:bcb-consistency}.
    A correct validator only executes certificates by following the sequence of monotonically increasing version numbers. 
    As a result, since $v$ executes $\cert$ before $\cert'$, it follows that $\version{\Obj}$ of $\cert$ is strictly lower than  $\version{\Obj}$ of $\cert'$.
    Similarly, since $v'$ executes $\cert'$ before $\cert$ it follows that $\version{\Obj}$ of $\cert'$ is strictly lower than  $\version{\Obj}$ of $\cert$.
    This implies that both $\cert$ and $\cert$ take as input $\Obj$ with the same version number. We finally note that correct validators only execute certificates valid for the current epoch 
    , thus $\cert$ and $\cert'$ share the same epoch number.
    As a result, $\cert$ and $\cert$ share the input $\Obj$ for the same version and epoch number and are thus conflicting.  This is a direct contradiction of \Cref{th:bcb-consistency}.
\end{proof}

\begin{lemma}[Shared Objects Sequential Execution] \label{th:sequencial-exec-shared}
    If two certificates $\cert$ and $\cert'$ both take as input the same shared object $\Obj$, all correct validators execute $\cert$ and $\cert'$ in the same order.
\end{lemma}
\begin{proof}
    Let's assume two correct validators $v$ and $v'$ execute in different orders $\cert$ and $\cert'$ taking the same input shared object $\Obj$. That is, $v$ executes $\cert$ then $\cert'$, and $v'$ executes $\cert'$ then $\cert$.
    We note $\transactiondigest = \txdigest{\cert}$, $\transactiondigest' = \txdigest{\cert'}$, and call $\ObjID$ the identifier of $\Obj$.
    We argue this lemma by contradiction of \Cref{th:shared-locks-consistency}.
    First, we note that Lamport timestamps ensure correct validators always assign strictly increasing version numbers to shared objects.
    Since $v$ executes $\cert$ before $\cert'$, its $\sharedlockdb$ store holds the following two entries:
    \begin{eqnarray} \nonumber
        \sharedlockdb[(\transactiondigest, \ObjID)] &=& \Version \\ \nonumber
        \sharedlockdb[(\transactiondigest', \ObjID)] &=& \Version'
    \end{eqnarray}
    with $\Version < \Version'$.
    Similarly, since $v'$ executes $\cert'$ before $\cert$, its $\sharedlockdb$ store holds the following two entries: \begin{eqnarray} \nonumber
        \sharedlockdb[(\transactiondigest', \ObjID)] &=& \Version' \\ \nonumber
        \sharedlockdb[(\transactiondigest, \ObjID)] &=& \Version
    \end{eqnarray}
    with $\Version' < \Version$.
    This however means that the stores of $v$ and $v'$ conflict, which is a direct contradiction of \Cref{th:shared-locks-consistency}.
\end{proof}

\para{State transistions}
We finally show that correct authorities executing the same sequence of certified transitions end up in the same state.

\begin{lemma}[Objects Identifiers Uniqueness] \label{th:objects-id-uniqueness}
    No polynomial-time constrained adversary can create two objects with the same identifier $\objectid$ without two successful invocations of $\txexec{\cert}$ over the same certificate $\cert$.
\end{lemma}
\begin{proof}
    We argue this lemma by the construction of the object identifier $\objectid$. \sysname derives each objects identifier $\objectid$ by hashing the digest $\txdigest{\cert}$ of the certificate creating the object along with an index unique to each input object of $\cert$. The adversary thus needs to find a hash collision to generate the same $\objectid$ twice through the invocation of $\txexec{\cert}$ and $\txexec{\cert'}$, where $\cert \neq \cert'$.
\end{proof}

\begin{lemma}[Deterministic Execution] \label{th:deterministic-exec}
    Every correct validator executing the same set of certificates makes the same state transitions.
\end{lemma}
\begin{proof}
    Every certificate $\cert$ in the sequence is executed by calling the the Move VM to produce the set of the newly created or mutated objects $[\Obj_{out}]$. The determinism of the Move VM and the correctness of its type checker ensures that every correct validator calling $\txexec{\cert}$ with the same input $\cert$ produces the same $[\Obj_{out}]$.
    %
    %
\end{proof}

\para{\sysname safety}
We now prove the safety of \sysname using the previous lemmas.

\begin{theorem}[\sysname Safety] \label{th:safety}
    At the start of every epoch, all correct validators have the same state.
\end{theorem}
\begin{proof}
    We argue this property by induction. Assuming a history of $n-1$ epochs for which this property holds we consider epoch $n$.
    \Cref{th:owned-object-execution-set} and \Cref{th:shared-object-execution-set} prove that all correct validators will execute the same set of transactions by the start of epoch $n+1$.
    Then \Cref{th:sequencial-exec-owned} and \Cref{th:sequencial-exec-shared} show that correct validators can only execute those transactions in the same order (whenever order matter).
    Finally, \Cref{th:deterministic-exec} shows that the execution of those transactions causes the same state transition across all correct validators.
    As a result, every correct validator will have the same state at the start of epoch $n+1$. The inductive base case is argued by construction: all correct validators start in the same state during the first epoch (i.e., genesis).
\end{proof}

\para{Client-perceived safety} Clients consider a transaction $\transaction$ final if there exists an effect certificate $\ecert$ over $\transaction$. We thus show that the existence of $\ecert$ implies that $\transaction$ is never reverted (\Cref{th:checkpoint-inclusion} of \Cref{sec:security} shows that $\transaction$ will eventually be executed). Thus all final transactions will be in a checkpoint within the epoch.

\begin{theorem}[Client-Perceived Safety] \label{th:client-safety}
    If there exists an effect certificate $\ecert$ over a transaction $\transaction$, the execution of $\transaction$ is never reverted.
\end{theorem}
\begin{proof}
    Let's assume there exists an effect certificate $\ecert$ over a transaction $\transaction$ and that the execution of $\transaction$ is reverted.
    The execution of $\transaction$ is reverted if and only if $\transaction$ is not included in a checkpoint by the end of the epoch. However, correct validators only sign $\ecert$ after including $\transaction$ in the list of certificates to be sequenced and eventually included into a checkpoint $c$. By the liveness property of consensus within an epoch, a correct validator will eventually be able to sequence the certificate as long as the epoch is ongoing.
    %
    %
    The epoch ending before the certificate being sequenced and included in a checkpoint implies that a set of $f+1$ correct validators signed $\ecert$ and did not see it included in a checkpoint within the epoch, and a disjoint set of $f+1$ correct validators did not sign $\ecert$ and participated in the reconfiguration protocol to to move to the next epoch. This implies a total of $f+1 + f+1 + f = 3f+2 > 3f+1$ validators, hence a contradiction.
\end{proof}

\begin{theorem}[No Conflicts]
    No two conflicting effect certificates exist. That is, two different effect certificates sharing the same input objects and object version.
\end{theorem}
\begin{proof}
    Let's assume two conflicting effect certificates $\ecert$ and $\ecert'$ exist.
    We distinguish two (exhaustive) cases, (i) $\ecert$ and $\ecert'$ share an input owned object with the same version number, and (ii) $\ecert$ and $\ecert'$ share an input shared object with the same version number.
    Case (i) implies there must exist two conflicting certificates $\cert$ and $\cert$ (remember effect certificate are created by signing certificates). This is however a direct contradiction of \Cref{th:bcb-consistency}.
    Case (ii) implies that the $f+1$ correct validators who signed $\ecert$ persisted $\sharedlockdb[(\transactiondigest, \ObjID)] = \Version$, and the $f+1$ correct validators who signed $\ecert'$ persisted $\sharedlockdb[(\transactiondigest, \ObjID')] = \Version$. Since \Cref{th:shared-locks-consistency} ensures the shared lock store of correct validators do not conflict, the set of $f+1$ correct validators who signed $\cert$ is disjoint from the set of $f+1$ correct validators who signed $\cert'$. As a result, there must be a total of $f+1 + f+1 +f = 3f+2 > 3f+1$ validators, hence a contradiction.
\end{proof}

\input{}

\subsection{Validity} \label{sec:validity}
We show that correct \sysname validators only execute valid transactions. Validity holds unconditionally to the network assumption. The protocol described in \Cref{sec:core-protocol} ensures validity as long as the BFT assumption holds. However, our implementation ensures validity unconditionally because correct validators re-run all validity checks upon processing a certificate. Should the BFT assumption break, they would thus early-reject certificates over invalid transactions before even starting to process them.

\begin{lemma}[Valid Certificates] \label{th:valid-certificates}
    All certified transactions are valid with respect to the authorization rules relating to owned object (defined in \Cref{sec:system-model}).
\end{lemma}
\begin{proof}
    Certificates are signed by at least $2f+1$ validators, out of which at least $f+1$ are correct. Correct validators only sign a transaction after ensuring that the transaction is valid with respect to the authorization rules relating to owned object.
    As a result, no invalid transaction will ever be signed by a correct authority, and will thus never be certified.
\end{proof}

\begin{theorem}[\sysname Validity]
    State transitions at correct validators are in accordance with (i) the authorization rules relating to owned object (defined in \Cref{sec:system-model}), and (ii) the Move smart contract logic constraining valid state transitions on objects of defined types.
\end{theorem}
\begin{proof}
    State transitions at correct validators only happen if a valid certificate $\cert$ exists and is processed. Point (i) is thus directly proven by the application of \Cref{th:valid-certificates}. 
    Point (ii) is proven by noting that correct validators only apply a state transition after calling  $\txexec{\cert}$ which only produces $[\Obj_{out}]$ by calling the Move smart contract logic constraining valid state transitions on objects of defined types.
\end{proof}

\subsection{Liveness} \label{sec:liveness}
We prove the liveness of the \sysname protocol. We start by showing that correct users can always obtain a certificate over their valid transactions, even across epochs.

\begin{lemma}[Dependencies Availability] \label{th:dependencies-availability}
    Given a certificate $\cert$ a correct user can always retrieve all the dependencies (i.e. parents) of $\cert$.
\end{lemma}
\begin{proof}
    We argue this property by induction on the serialized retrieval of the direct parent certificates. Assuming a history of $n+1$ certificate dependencies for which this property holds, we consider certificate $n$ noted $\cert$.
    $\cert$ is signed by $2f+1$ validators, out of which at least $f+1$ are correct. Correct validators only sign a transaction after ensuring they hold all its input objects.
    This means that $f+1$ correct validators have executed (and persisted) certificate $n-1$ that created the inputs of $\cert$. A correct user can thus query any of those $f+1$ correct validators for certificate $n-1$.
    The inductive base case assumes that the first dependency of every certificate is a fixed genesis (which we ensure axiomatically).
\end{proof}

\begin{lemma}[Certificate Creation] \label{th:certificate-creation}
    A correct user can obtain a certificate $\cert$ over a valid transaction $\transaction$.
\end{lemma}
\begin{proof}
    A correct validator always signs a transaction $\transaction$ if it passes the 4 checks of function \textsc{ProcessCert} of Algorithm 1 (see the long version of the paper~\cite{sui-lutris}). The first check (1.1) ensures that all objects referenced by the transaction exist, the second check (1.2) ensures the transaction is valid with respect to the authorization rules relating to owned objects, the third check (1.3) ensures no concurrent transaction is concurrently accessing the same input objects, and the fourth check (1.4) ensures that all these objects can be locked.

    \Cref{th:dependencies-availability} proves that a correct user can always ensure check (1.1) passes by providing all the transaction's dependencies the validator missed.
    Correct transactions always pass check (1.2).
    Check (1.3) always passes for the first copy of $\transaction$ received by the validator (at any given time).
    Finally correct users do not equivocate. Thus $\transaction$ is the first and only transaction referencing its owned object, and always passes check (1.4).
    As a result, if $\transaction$ is disseminated to $2f+1$ correct validators by a correct user, they will eventually all return a signature $\txsign$ to the user. The user then aggregates those $\txsign$ into a certificate $\cert$ over $\transaction$.
\end{proof}

\begin{lemma}[Certificate Renewal] \label{th:renew-certificate}
    A correct user holding a certificate over transaction $\transaction$ for an old epoch $e$ that did not finalize in $e$ can get a new certificate over $\transaction$ for the current epoch $e'$.
\end{lemma}
\begin{proof}
    A correct user holding a certificate over transaction $\transaction$ for an old epoch $e$ can re-submit $\transaction$ to $2f+1$ correct validators to obtain a new certificate for the current epoch $e'$. Indeed, correct validators sign $\transaction$ if it passes all 4 checks of function \textsc{ProcessTx} of Algorithm 1 (see the long version of the paper~\cite{sui-lutris}) like in \Cref{th:certificate-creation}.
    If the validator did not already execute $\transaction$, the correct user can ensure check (1.1) passes by providing all the transaction's dependencies the validator missed (\Cref{th:dependencies-availability}).
    Check (1.2) and (1.3) will pass exactly as described in \Cref{th:certificate-creation}.
    Finally, check (1.4) passes since correct users do not attempt equivocation during epoch $e'$ and correct validators drop the store $\ownedlockdb$ upon epoch change (and thus $\ownedlockdb[\objectkey] == \text{None}$, for every input of $\transaction$).
    Thus, if $\transaction$ is disseminated to $2f+1$ correct validators by a correct user, they will eventually all return a signature $\txsign$ to the user. The user then aggregates those $\txsign$ into a certificate $\cert$ over $\transaction$.
\end{proof}

The existence of a certificate implies that every owned object used as input of a certified transaction is locked for a particular version number. We now prove that all the shared objects of the certificate are also eventually locked for a version number.

\begin{lemma}[Shared Locks Availability] \label{th:shared-locks-availability}
    A correct user can always ensure that all correct validators eventually assign shared locks to all shared objects of a valid transaction $\transaction$.
\end{lemma}
\begin{proof}
    \Cref{th:certificate-creation} shows that a correct user can always assemble a certificate $\cert$ over a valid transaction $\transaction$.
    The correct user can then forward the $\cert$ to an honest authority who submits it to the consensus engine. By the liveness property of the consensus, $\cert$ is eventually sequenced by all correct validators. When $\cert$ is sequenced, correct validators assign locks to all shared objects of $\cert$.
    %
\end{proof}

We finally show that the existence of a certificate ensures the transactions of a correct user are eventually included in a checkpoint, and thus eventually executed.

\begin{lemma}[Effect Certificates Availability] \label{th:effect-dependencies-availability}
    A correct user can always ensure an effect certificate $\ecert$ over transaction $\transaction$ will eventually exist if a certificate $\cert$ over $\transaction$ exists.
\end{lemma}
\begin{proof}
    \lef{@ALberto:Same for this}
    A correct validator signs an effect $\esign$ if it passes all 3 checks of function \textsc{ProcessCert} of Algorithm 4 (see the long version of the paper). The first check (4.1) ensures that the certificate was issued during the current epoch, the second check (4.2) ensures that all dependencies of the certificate are available, and the third check (4.3) ensures that all shared objects of the certificate are (i) known to the validators and (ii) locked for the next version number.

    A correct user can ensure that check (4.1) passes by either providing the validator with the certificate $\cert$ during the same epoch of its creation or by re-creating a certificate for the current epoch (\Cref{th:renew-certificate}).
    A correct user can ensure check (4.2) passes by providing all the certificate's dependencies the validator missed (\Cref{th:dependencies-availability}).
    \Cref{th:shared-locks-availability} ensures that a correct user can make correct validators assign shared locks to all shared objects of a certificate $\cert$, thus validating the first part of check (4.3). It can then ensure the second part of check (4.3) succeeds by providing the validator with all dependencies it missed.
    As a result, a correct user can collect at least $2f+1$ effects $\esign$ over $\transaction$ and assemble them into an effect certificate $\ecert$.
\end{proof}

\begin{lemma}[Checkpoint Inclusion] \label{th:checkpoint-inclusion}
    If an effect certificate over transaction $\transaction$ exists within an epoch, $\transaction$ will be included in a checkpoint within the same epoch.
\end{lemma}
\begin{proof}
    If an effect certificate $\ecert$ exists, at least $f+1$ correct validators executed its corresponding transaction $\transaction$.  When correct validators execute a transaction they include it in the list of certificates to sequence and checkpoint (\Cref{sec:checkpoint-protocol}). Since $f+1$ correct validators are also needed to close the epoch, and a correct validator will not do so until it witnesses all listed certificates being sequenced, and by the liveness of consensus within an epoch, it follows that eventually the certificate will be sequenced (similar to \Cref{th:client-safety}). Since all certificates on which a certificate depends must also have been executed (in case of owned objects) or sequenced and executed (in case of shared objects) before an honest validator executes the transaction and signs it, it follows that if an $\ecert$ exists then an $\ecert$ for all dependencies will also exist and also be eventually sequenced. Since the certificate and all its causal dependencies will eventually be sequenced within the epoch, they will be included in a checkpoint within the epoch.
\end{proof}

\begin{lemma}[Checkpoint Execution] \label{th:checkpoint-execution}
    All correct validators eventually execute all transactions included in all checkpoints.
\end{lemma}
\begin{proof}
    Correct validators assemble checkpoints out of certificates sequenced by consensus. By the liveness property of the consensus protocol, all correct validators can eventually sequence all the certificates they are executed (or observe others do so) and assemble them into checkpoints. We conclude the proof by noting that correct validators execute all transactions within all checkpoints they assemble.
\end{proof}

\begin{theorem}[\sysname Liveness]
    A correct user can always ensure its transaction $\transaction$ will eventually be finalized. That is, all correct validators execute it and never revert it.
\end{theorem}
\begin{proof}
    \Cref{th:certificate-creation} ensures that a correct user can eventually obtain a certificate $\cert$ over their valid transaction $\transaction$.
    \Cref{th:effect-dependencies-availability} then ensures the user can get an effect certificate $\ecert$ using $\cert$.
    \Cref{th:checkpoint-inclusion} shows that the existence of $\ecert$ implies the transaction is eventually included in the a checkpoint.
    Finally \Cref{th:checkpoint-execution} shows that all transactions included in all checkpoints are executed.
    To conclude the proof we note that the execution of transactions included in checkpoints is never reverted (\Cref{sec:reconfiguration-protocol}).
\end{proof}

\begin{theorem}[Client-Perceived Starvation Freedom] \label{th:client-starvation}
    Let's assume two correct validators respectively set $\ownedlockdb[\objectkey] = \sign{\transaction_1}$ and $\ownedlockdb[\objectkey] = \sign{\transaction_2}$ (with $\transaction_1 \neq \transaction_2$) during epoch $e$.
    A correct user can eventually obtain an effect certificate over transaction $
        \transaction'$ accessing $\objectkey$ at epoch $e' > e$.
\end{theorem}
\begin{proof}
    All owned objects locked by transactions at epoch $e$ are freed upon entering epoch $e+1$ (see \Cref{sec:reconfiguration}); that is, correct validators drop all $\ownedlockdb[\cdot]$ upon epoch change.
    Correct validators thus sign a correct transaction $\transaction'$ accessing $\objectkey$ at epoch $e' > e$ submitted by correct clients (who do not equivocate). \Cref{th:certificate-creation} then ensures the client eventually obtains a certificate over $\transaction'$ and \Cref{th:effect-dependencies-availability} ensures the client eventually obtains an effect certificate over $\transaction'$.
\end{proof}

\section{Implementation} \label{sec:implementation}

We implement a networked multi-core \sysname validator in Rust forking the FastPay~\cite{fastpay-code},  Narwhal~\cite{narwhal-code}, and Bullshark~\cite{bullshark-code} projects. It uses Tokio~\cite{tokio-code} for asynchronous networking, fastcrypto~\cite{fastpay-code} for elliptic curve based  signatures. Data-structures are persisted using RocksDB~\cite{rocksdb-code}. We use QUIC~\cite{quinn-code} to achieve reliable authenticated point-to-point channels. The implementation of \sysname is around 30 kLOC and over 10 kLOC of tests.
Contrarily to most prototypes, our implementation is production-ready and fully-featured. It runs at the heart of a major new blockchain (integrated in 350 kLOC) mainnet.
\extendedBragging
%
We are open sourcing our implementation of \sysname\footnote{
  \ifdefined\cameraReady
    \url{https://github.com/mystenlabs/sui}
  \else
    Code available and deployed in production but link omitted for review.
  \fi
}.

\section{Evaluation} \label{sec:evaluation}

We evaluate the throughput and latency of our implementation of \sysname through experiments on AWS.
We particularly aim to demonstrate the following claims.
\textbf{C1} \sysname achieves high throughput despite its consensus-intensive checkpointing mechanism.
\textbf{C2} \sysname finalizes owned-objects transactions with sub-second latency (in the WAN) for small and medium committee sizes or when the system is under low load.
\textbf{C3} \sysname is robust when some parts of the system inevitably crash-fail.
\textbf{C4} \sysname validators are capable of quickly recovering after crashes without visible performance impact.
\textbf{C5} \sysname's epoch change mechanism only causes a small disruption.
%
%
Experimental evaluation of BFT protocols in the presence of Byzantine faults is an open research question~\cite{bano2020twins} but \Cref{sec:safety} and \Cref{sec:security} provide detailed safety and liveness proofs.

To demonstrate those claims we compare the fast path of \sysname with an implementation of Bullshark that uses the same execution engine and consensus protocol as \sysname. This means that although Bullshark is used as a baseline it also represents the actual performance of \sysname when deployed only with a shared objects workload. We use our implementation of Bullshark, which is an extension of the original codebase, for three reasons. First, the original Bullshark protocol only measures throughput by counting the number of bytes committed; as a result the throughput reported did not correspond to the performance of a real (production-ready) system because there was no execution or state to be updated (i.e., Bullshark is a BAB~\cite{cristian1995atomic} and not an SMR~\cite{schneider1990implementing}). Second, as a result of counting bytes there was no duplicate suppression, resulting in a potential $O(n)$ duplicate transactions. In our implementation we report the goodput of the system, i.e., the number of distinct certificates finalized by the system. Finally, due to lack of state it was trivial to DoS the original Bullshark by sending invalid and malformed transactions such that the goodput would drop to zero.
Instead, in our experiments we ensures that all transactions forwarded to the consensus engine use valid gas objects that are spent regardless of whether the transaction bytes are valid.

\subsection{Experimental Setup}
We deploy a fully-featured \sysname testbed on AWS, using \texttt{m5d.8xlarge} instances across 13 different AWS regions: N. Virginia (us-east-1), Oregon (us-west-2), Canada (ca-central-1), Frankfurt (eu-central-1), Ireland (eu-west-1), London (eu-west-2), Paris (eu-west-3), Stockholm (eu-north-1), Mumbai (ap-south-1), Singapore (ap-southeast-1), Sydney (ap-southeast-2), Tokyo (ap-northeast-1), and Seoul (ap-northeast-2).
Validators are distributed across those regions as equally as possible. Each machine provides 10Gbps of bandwidth, 32 virtual CPUs (16 physical core) on a 2.5GHz, Intel Xeon Platinum 8175, 128GB memory, and runs Linux Ubuntu server 22.04. \sysname persists all data on the NVMe drives provided by the machine (rather than the root partition). We select these machines because they provide decent performance and are in the price range of `commodity servers'.

In all graphs of the paper, each data point is the average of the latency of all transactions of the run, and the error bars represent one standard deviation (errors bars are sometimes too small to be visible on the graph). 
We instantiate several geo-distributed benchmark clients submitting transactions at a fixed rate for a duration of 10 minutes; unless specified otherwise each benchmark client submits at most 350 tx/s and the number of clients thus depends on the desired input load.

When referring to \emph{latency}, we mean the time elapsed from when the client submits the transaction to when it assembles an effect certificate (and the transaction has reached settlement, see \Cref{sec:overview}). The latency of \sysname shared-object transaction finality, are lower than the ones of Bullshark in our experiments: \sysname reaches transaction finality before execution with a latency equal to the owned-object transaction. The Bullshark latency is similar to \sysname's shared-object settlement latency instead (see \Cref{sec:overview}). When referring to \emph{throughput}, we mean the number of effect certificates over \emph{distinct} transactions over the entire duration of the run. Since \sysname also accepts bundles of multiple transactions (called \emph{programmable transaction blocks}, see \Cref{sec:common-case}) as input, we report throughput as certificates per second (cert/s) to denote distinct unrelated transaction throughput or how many users can concurrently be served.

Transactions processed by \sysname are payment transfers and transactions processed by Bullshark are increments of a shared counter which simulates the sequence number of a shared account.
%
\ifdefined\cameraReady
    We benchmark the version of the codebase \texttt{mainnet-v1.4.3}\footnote{
        \url{https://github.com/asonnino/sui/tree/sui-lutris} (commit \texttt{7f3d922})
    } deployed on the Sui mainnet in July 2023. We open-source all orchestration scripts, benchmarking scripts, and measurements data to enable reproducible evaluation results\footnote{
        \url{https://github.com/asonnino/sui-paper/tree/main/data}
    }. The extended version of the paper~\cite{sui-lutris} provides a tutorial to reproduce our experiments.
\fi

\subsection{Benchmark in the Common Case} \label{sec:common-case}
\Cref{fig:latency-common-case} compares the performance of \sysname with the baseline Bullshark when running with 10, 50, and 100 non-faulty validators. The lower part of the figure provides another view of the same data by showing the maximum achievable throughput while keeping the latency below 0.5 seconds and 5 seconds (the system's SLA).

Bullshark finalizes 4,000 cert/s with a committee size of 10 or 50 and 3,500 cert/s with a committee size of 100. In comparison, \sysname finalizes 5,000 cert/s with a committee size of 10 or 50, and 4,000 cert/s with a committee size of 100. In all cases, \sysname's throughput outperforms Bullshark's despite its checkpointing mechanism ( \Cref{sec:checkpoint}) that sequences all certificates. This observation validates our claim \textbf{C1}.
Regardless of the committee size, Bullshark's latency is about 3 seconds while \sysname's latency is less than 0.5 seconds -- corresponding to a 6x latency improvement. This observation validate our claim \textbf{C2}.

\begin{figure}[t]
    \includegraphics[width=\columnwidth]{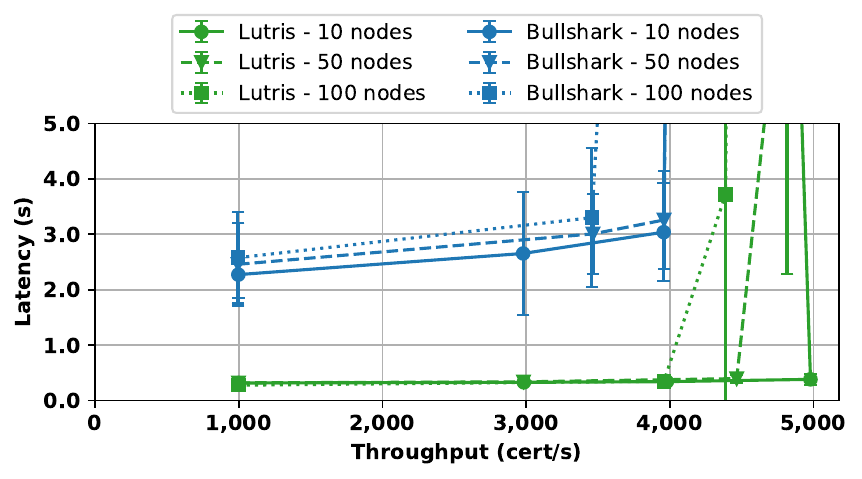}
    \includegraphics[width=\columnwidth]{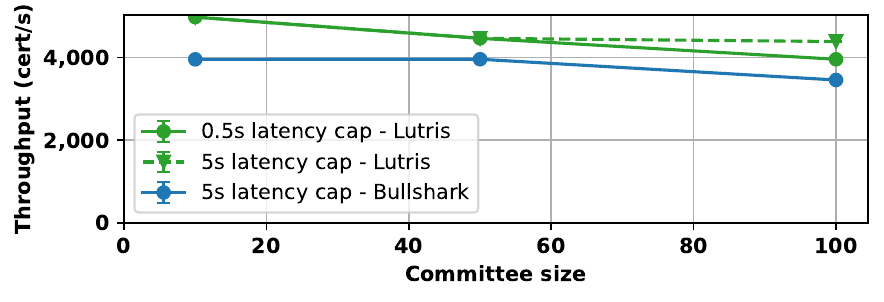}
    \vspace{-0.5cm}
    \caption{\sysname and Bullshark WAN latency-throughput with 10, 50, and 100 validators (no faults).}
    \label{fig:latency-common-case}
\end{figure}

\noindent \textbf{Programmable transaction blocks.}
\sysname allows a client to `bundle' multiple transactions into a single system transaction, and signing only the bundle rather than each individual transaction. We call such bundle \emph{programmable transaction block} (PTB). This is useful for large exchanges and corporate entities submitting numerous transactions on behalf of their users. Bundling transactions reduces the number of messages exchanged between the client and the system and greatly reduces the cost of signature verification and can be seen as equivalent to the batching strategies consensus protocols take to increase their throughput.

\Cref{fig:bundle-latency} shows the performance of a 100-validator deployment of \sysname executing a payload composed of PTBs of 100 transactions.
The graphs shows a peak at 150,000 ops/s indicating that \sysname can process 1,500 bundles of 100 transactions per second while keeping latency below 0.5 seconds.

\begin{figure}[t]
    \vspace{-0.25cm}
    \includegraphics[width=\columnwidth]{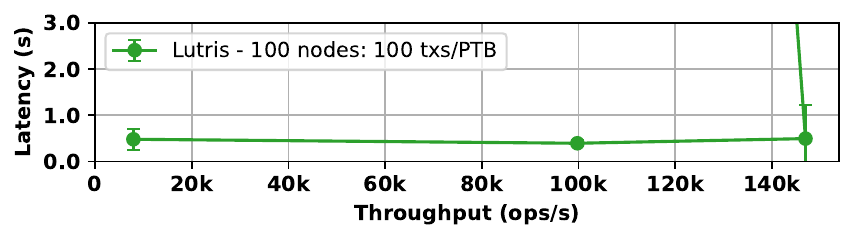}
    \vspace{-0.5cm}
    \caption{\sysname latency-throughput with bundles of 100 transactions per programmable transaction block (PTB); 100 validators, no faults.}
    \label{fig:bundle-latency}
    \vspace{-0.5cm}
\end{figure}

\subsection{Benchmark with Faults}
\Cref{fig:latency-faults} compares the performance of \sysname with the baseline Bullshark for a 10-validator deployment when the system experiences (crash-)faults; after running without faults for one minute, 1 and 3 validators crash.
The figure shows that Bullshark can finalize 3,500 cert/s in about 5 seconds and 3,000 cert/s in about 7.5 seconds when respectively 1 and 3 validators crashes. In contrast, \sysname is largely unaffected by validator's crashes: it can still finalize over 4,000 cert/s with a latency of less than 0.5 seconds. We thus observe that \sysname provides up to 15x latency reduction when the system experiences (crash-)faults. This observation validates our claim \textbf{C3}.
\Cref{fig:crash-recovery} shows the throughput and latency of \sysname and Bullshark when 3 validators are crashing and recovering. The plots are divided in 5 zones by vertical black lines; no validators are crashed in the first zone; then respectively 1, 2, and 3 validators are crashed in the 2nd, 3rd, and 4th zone; and all validators recover in the last zone. The systems are submitted to a constant load of 3,000 cert/s (regardless of the number of faults). Each point on the graphs is the average metric observed by the clients (averaged over all clients and with a 15-seconds window).
As expected, the throughput of Bullshark slightly degrades (barely visible) and its latency increases when the number of crash-faults increases. The 5th zone shows that performance starts recovering when all validators recover. The slight delay in recovery is due to the overhead required by our orchestration to apply changes to the deployment environment. In contrast, \sysname is largely unaffected by crash-recovering validators, validating claim \textbf{C4}.

\begin{figure}[t]
    \includegraphics[width=\columnwidth]{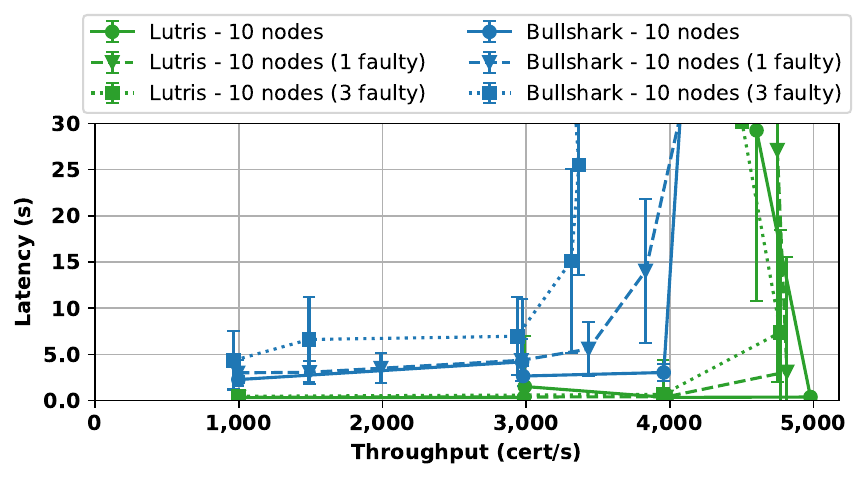}
    \vspace{-0.5cm}
    \caption{\sysname and Bullshark WAN latency-throughput with 20 validators (1, 3, and 6 faults).}
    \label{fig:latency-faults}
\end{figure}

\begin{figure}[t]
    \vspace{-0.45cm}
    \centering
    \includegraphics[width=\columnwidth]{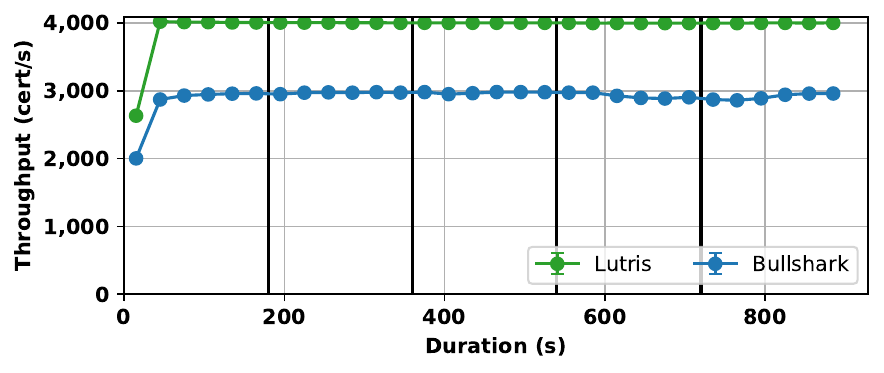}
    \vspace{-0.1cm}
    \includegraphics[width=\columnwidth]{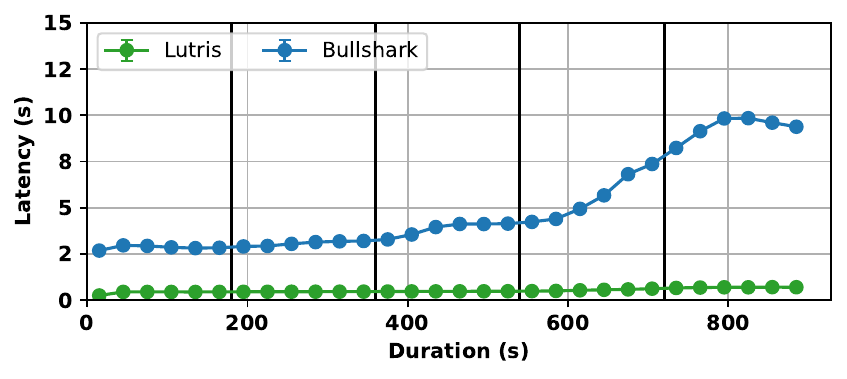}
    \vspace{-0.7cm}
    \caption{Performance of a 10-validators committee when up to 3 validators crash and recover.}
    \label{fig:crash-recovery}
    \vspace{-0.45cm}
\end{figure}

\subsection{Stability during Epoch Changes}

\Cref{fig:epoch-change} shows the throughput and latency over time of 10-validator deployments of \sysname and Bullshark. The systems are submitted to a constant load of 3,000 cert/s for about 35 minutes during which the systems undergo 3 epoch changes (one every 10 minutes, indicated by black vertical lines). Noting that the performance of \sysname (and Bullshark) are largely unaffected by epoch changes validates claim \textbf{C5}.

\begin{figure}[t]
    \includegraphics[width=\columnwidth]{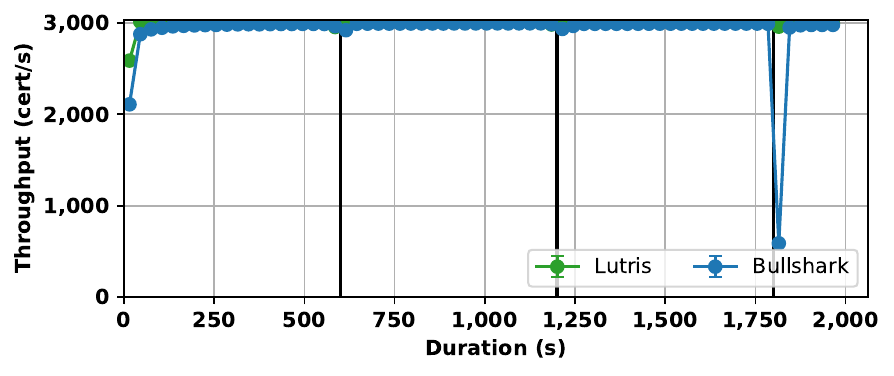}
    \vspace{-0.2cm}
    \includegraphics[width=\columnwidth]{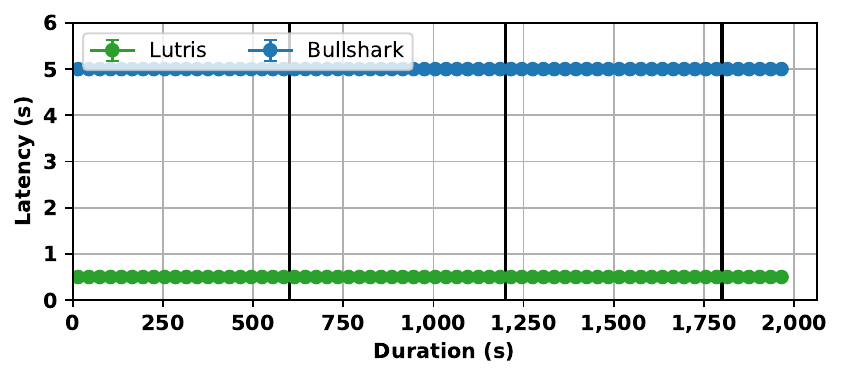}
    \vspace{-0.5cm}
    \caption{Performance of a 10-validators committee during epoch changes.}
    \label{fig:epoch-change}
    \vspace{-0.5cm}
\end{figure}
\section{Related and Future Work}
\sysname is the first deployed secure smart-contract platform that provides sub-second finality for distributed ledger transactions without compromising the expressiveness and the performance of state-of-the-art consensus protocols such as Bullshark.
To achieve this it combines said state-of-the-art consensus protocols with ideas from the FastPay~\cite{BaudetDS20} low-latency settlement systems to gain the ability to operate on arbitrary objects through user-defined smart contracts, and with a delegated proof-of-stake committee~\cite{sok-consensus}.

The \sysname owned object path is based on Byzantine consistent broadcast~\cite{cachin2011introduction}. Previous works suggested using this weaker primitive to build payment systems~\cite{GuerraouiKMPS19,DBLP:journals/corr/abs-1812-10844,astro} but lack an integration with a consensus path making it both unsuitable to run indefinitely (no garbage-collection or reconfiguration) as well as limited functionality (only payments) and usability (client-side bugs result in permanent loss of funds).

Other systems similar to \sysname are Astro~\cite{astro} and ABC/CoD~\cite{abc,sliwinski2022consensus}.
Astro relies on an eager implementation of \emph{Byzantine reliable broadcast}~\cite{cachinBook} which achieves \emph{totality}~\cite{cachinBook} without relying on an external synchronizer at the cost of higher communication in the common case. Additionally, Astro is designed as a standalone payment system but does not handle checkpointing or reconfiguration.
Similar to Astro and FastPay, ABC~\cite{abc} proposes a relaxed notion of consensus where termination is only guaranteed for honest senders, this however is quite disruptive for the client experience as simple mistakes cause complete loss of access to user assets.

\sysname's fast path operates with a latency of 2RTTs. This is higher than the 1RTT of FastPay and Astro but enables secure reconfiguration and the capability to unlock equivocated user assets. The RTTs of the shared object path of \sysname match those of its underlying black-box consensus mechanism. The mixed-objects path requires one extra RTT for certification but this step does not only serve as a locking mechanism between the fast path and consensus. It also serves as a pre-filter for spam transactions. BAB protocols do not handle transaction prevalidation on their own, which means that malicious clients can submit junk transactions and flood validators with duplicates (\Cref{sec:evaluation}). Existing blockchains perform this pre-filtering through a gossip layer, which increases (unquantifiable) latency. We instead added a certification step that both implements our locking mechanism and thwarts these attacks, ensuring that consensus channels execute only transactions that consume gas and yield profits for validators.

\sysname integration with a consensus protocol able to keep up with checkpointing the throughput the fast path is able to execute. For this reason, we chose the use of Narwhal-Bullshark~\cite{narwhal,bullshark} in a variant without asynchronous fallback~\cite{spiegelman2022bullshark}, due to their reported and observed high performance.
Similarly, \sysname requires integration with a deterministic execution engine in order to provide end-to-end application level semantics instead of simply ordering bytes.
Nevertheless, our design is modular~\cite{cohen2022proof} and can interface both with any total-ordering protocol~\cite{malkhi2023hotstuff, gelashvili2022jolteon, kokoris2018omniledger, al2017chainspace, castro1999practical,sonnino2020replay} as well as with any deterministic execution engine~\cite{gelashvili2023block, zhang2022ethereum}.

\ifdefined\cameraReady
    \section*{Acknowledgments}
This work is funded by MystenLabs. We thank the Mysten Labs Engineering teams for valuable feedback broadly, and specifically Dmitry Perelman and Todd Fiala for
managing the implementation effort. A number of folks contributed to specific aspects of the implementation of \sysname (amongst many other contributions to the overall blockchain):
%
Francois Garillot, Laura Makdah,
%
Mingwei Tian, Andrew Schran, Sadhan Sood and William Smith implemented and optimized aspects of both \sysname and Narwhal / Bullshark consensus;
%
Alonso de Gortari oversaw the crypto economics of the blockchain, and Emma Zhong, Ade Adepoju, Tim Zakia and Dario Russi designed and implemented staking and gas mechanisms.
Adam Welc designed several Move tools and provided great feedback on the manuscript.
We also extend our thanks to Patrick Kuo, Ge Gao, Chris Li, and Arun Koshy for their work on the \sysname SDK, clients, and RPC layer;
Kostas Chalkias, Jonas Lindstrøm, and Joy Wang built cryptographic components.

\fi

\bibliographystyle{ACM-Reference-Format}
\bibliography{references}


\begin{thebibliography}{44}


\ifx \showCODEN    \undefined \def \showCODEN     #1{\unskip}     \fi
\ifx \showDOI      \undefined \def \showDOI       #1{#1}\fi
\ifx \showISBNx    \undefined \def \showISBNx     #1{\unskip}     \fi
\ifx \showISBNxiii \undefined \def \showISBNxiii  #1{\unskip}     \fi
\ifx \showISSN     \undefined \def \showISSN      #1{\unskip}     \fi
\ifx \showLCCN     \undefined \def \showLCCN      #1{\unskip}     \fi
\ifx \shownote     \undefined \def \shownote      #1{#1}          \fi
\ifx \showarticletitle \undefined \def \showarticletitle #1{#1}   \fi
\ifx \showURL      \undefined \def \showURL       {\relax}        \fi
\providecommand\bibfield[2]{#2}
\providecommand\bibinfo[2]{#2}
\providecommand\natexlab[1]{#1}
\providecommand\showeprint[2][]{arXiv:#2}

\bibitem[Al-Bassam et~al\mbox{.}(2018)]%
        {al2017chainspace}
\bibfield{author}{\bibinfo{person}{Mustafa Al-Bassam}, \bibinfo{person}{Alberto Sonnino}, \bibinfo{person}{Shehar Bano}, \bibinfo{person}{Dave Hrycyszyn}, {and} \bibinfo{person}{George Danezis}.} \bibinfo{year}{2018}\natexlab{}.
\newblock \showarticletitle{Chainspace: A sharded smart contracts platform}. In \bibinfo{booktitle}{\emph{NDSS}}.
\newblock


\bibitem[Bano et~al\mbox{.}(2019)]%
        {sok-consensus}
\bibfield{author}{\bibinfo{person}{Shehar Bano}, \bibinfo{person}{Alberto Sonnino}, \bibinfo{person}{Mustafa Al-Bassam}, \bibinfo{person}{Sarah Azouvi}, \bibinfo{person}{Patrick McCorry}, \bibinfo{person}{Sarah Meiklejohn}, {and} \bibinfo{person}{George Danezis}.} \bibinfo{year}{2019}\natexlab{}.
\newblock \showarticletitle{{SoK}: Consensus in the age of blockchains}. In \bibinfo{booktitle}{\emph{ACM Advances in Financial Technologies (AFT)}}.
\newblock


\bibitem[Bano et~al\mbox{.}(2020)]%
        {bano2020twins}
\bibfield{author}{\bibinfo{person}{Shehar Bano}, \bibinfo{person}{Alberto Sonnino}, \bibinfo{person}{Andrey Chursin}, \bibinfo{person}{Dmitri Perelman}, \bibinfo{person}{Zekun Li}, \bibinfo{person}{Avery Ching}, {and} \bibinfo{person}{Dahlia Malkhi}.} \bibinfo{year}{2020}\natexlab{}.
\newblock \showarticletitle{Twins: {BFT} systems made robust}. In \bibinfo{booktitle}{\emph{OPODIS}}.
\newblock


\bibitem[Baudet et~al\mbox{.}(2020)]%
        {BaudetDS20}
\bibfield{author}{\bibinfo{person}{Mathieu Baudet}, \bibinfo{person}{George Danezis}, {and} \bibinfo{person}{Alberto Sonnino}.} \bibinfo{year}{2020}\natexlab{}.
\newblock \showarticletitle{FastPay: High-Performance Byzantine Fault Tolerant Settlement}. In \bibinfo{booktitle}{\emph{AFT}}.
\newblock


\bibitem[Baudet et~al\mbox{.}(2022)]%
        {zef}
\bibfield{author}{\bibinfo{person}{Mathieu Baudet}, \bibinfo{person}{Alberto Sonnino}, \bibinfo{person}{Mahimna Kelkar}, {and} \bibinfo{person}{George Danezis}.} \bibinfo{year}{2022}\natexlab{}.
\newblock \showarticletitle{Zef: Low-latency, Scalable, Private Payments}.
\newblock \bibinfo{journal}{\emph{Arxiv}} (\bibinfo{year}{2022}).
\newblock


\bibitem[Blackshear et~al\mbox{.}(2019a)]%
        {move_white}
\bibfield{author}{\bibinfo{person}{Sam Blackshear}, \bibinfo{person}{Evan Cheng}, \bibinfo{person}{David~L. Dill}, \bibinfo{person}{Victor Gao}, \bibinfo{person}{Ben Maurer}, \bibinfo{person}{Todd Nowacki}, \bibinfo{person}{Alistair Pott}, \bibinfo{person}{Shaz Qadeer}, \bibinfo{person}{Ra\ in}, \bibinfo{person}{Dario Russi}, \bibinfo{person}{Stephane Sezer}, \bibinfo{person}{Tim Zakian}, {and} \bibinfo{person}{Runtian Zhou}.} \bibinfo{year}{2019}\natexlab{a}.
\newblock \bibinfo{title}{Move: A Language With Programmable Resources}.
\newblock \bibinfo{howpublished}{\url{https://move-book.com}}.
\newblock


\bibitem[Blackshear et~al\mbox{.}(2019b)]%
        {move}
\bibfield{author}{\bibinfo{person}{Sam Blackshear}, \bibinfo{person}{Evan Cheng}, \bibinfo{person}{David~L Dill}, \bibinfo{person}{Victor Gao}, \bibinfo{person}{Ben Maurer}, \bibinfo{person}{Todd Nowacki}, \bibinfo{person}{Alistair Pott}, \bibinfo{person}{Shaz Qadeer}, \bibinfo{person}{Dario~Russi Rain}, \bibinfo{person}{Stephane Sezer}, {et~al\mbox{.}}} \bibinfo{year}{2019}\natexlab{b}.
\newblock \showarticletitle{Move: A language with programmable resources}.
\newblock \bibinfo{journal}{\emph{Libra Assoc}} (\bibinfo{year}{2019}).
\newblock


\bibitem[Blackshear et~al\mbox{.}(2023)]%
        {sui-lutris}
\bibfield{author}{\bibinfo{person}{Sam Blackshear}, \bibinfo{person}{Andrey Chursin}, \bibinfo{person}{George Danezis}, \bibinfo{person}{Anastasios Kichidis}, \bibinfo{person}{Lefteris Kokoris-Kogias}, \bibinfo{person}{Xun Li}, \bibinfo{person}{Mark Logan}, \bibinfo{person}{Ashok Menon}, \bibinfo{person}{Todd Nowacki}, \bibinfo{person}{Alberto Sonnino}, {et~al\mbox{.}}} \bibinfo{year}{2023}\natexlab{}.
\newblock \bibinfo{title}{Sui Lutris: A Blockchain Combining Broadcast and Consensus}.
\newblock \bibinfo{howpublished}{\url{https://arxiv.org/abs/2310.18042}}.
\newblock


\bibitem[Boneh et~al\mbox{.}(2001)]%
        {BonehLS01}
\bibfield{author}{\bibinfo{person}{Dan Boneh}, \bibinfo{person}{Ben Lynn}, {and} \bibinfo{person}{Hovav Shacham}.} \bibinfo{year}{2001}\natexlab{}.
\newblock \showarticletitle{Short Signatures from the Weil Pairing}. In \bibinfo{booktitle}{\emph{Theory and Application of Cryptology and Information Security (ASIACRYPT)}}.
\newblock


\bibitem[Cachin et~al\mbox{.}(2011a)]%
        {cachin2011introduction}
\bibfield{author}{\bibinfo{person}{Christian Cachin}, \bibinfo{person}{Rachid Guerraoui}, {and} \bibinfo{person}{Lu{\'\i}s Rodrigues}.} \bibinfo{year}{2011}\natexlab{a}.
\newblock \bibinfo{booktitle}{\emph{Introduction to reliable and secure distributed programming}}.
\newblock \bibinfo{publisher}{Springer Science \& Business Media}.
\newblock


\bibitem[Cachin et~al\mbox{.}(2011b)]%
        {cachinBook}
\bibfield{author}{\bibinfo{person}{Christian Cachin}, \bibinfo{person}{Rachid Guerraoui}, {and} \bibinfo{person}{Lu{\'\i}s Rodrigues}.} \bibinfo{year}{2011}\natexlab{b}.
\newblock \bibinfo{booktitle}{\emph{Introduction to reliable and secure distributed programming}}.
\newblock \bibinfo{publisher}{Springer Science \& Business Media}.
\newblock


\bibitem[Castro et~al\mbox{.}(1999)]%
        {castro1999practical}
\bibfield{author}{\bibinfo{person}{Miguel Castro}, \bibinfo{person}{Barbara Liskov}, {et~al\mbox{.}}} \bibinfo{year}{1999}\natexlab{}.
\newblock \showarticletitle{Practical byzantine fault tolerance}. In \bibinfo{booktitle}{\emph{Usenix OSDI}}.
\newblock


\bibitem[Celestia(2022)]%
        {celestia}
\bibfield{author}{\bibinfo{person}{Celestia}.} \bibinfo{year}{2022}\natexlab{}.
\newblock \bibinfo{title}{{The first modular blockchain network}}.
\newblock \bibinfo{howpublished}{\url{https://celestia.org}}.
\newblock


\bibitem[Cohen et~al\mbox{.}(2023)]%
        {cohen2022proof}
\bibfield{author}{\bibinfo{person}{Shir Cohen}, \bibinfo{person}{Guy Goren}, \bibinfo{person}{Lefteris Kokoris-Kogias}, \bibinfo{person}{Alberto Sonnino}, {and} \bibinfo{person}{Alexander Spiegelman}.} \bibinfo{year}{2023}\natexlab{}.
\newblock \showarticletitle{Proof of Availability \& Retrieval in a Modular Blockchain Architecture}.
\newblock \bibinfo{journal}{\emph{Financial Cryptography and Data Securty (FC)}} (\bibinfo{year}{2023}).
\newblock


\bibitem[Collins et~al\mbox{.}(2020)]%
        {astro}
\bibfield{author}{\bibinfo{person}{Daniel Collins}, \bibinfo{person}{Rachid Guerraoui}, \bibinfo{person}{Jovan Komatovic}, \bibinfo{person}{Petr Kuznetsov}, \bibinfo{person}{Matteo Monti}, \bibinfo{person}{Matej Pavlovic}, \bibinfo{person}{Yvonne-Anne Pignolet}, \bibinfo{person}{Dragos-Adrian Seredinschi}, \bibinfo{person}{Andrei Tonkikh}, {and} \bibinfo{person}{Athanasios Xygkis}.} \bibinfo{year}{2020}\natexlab{}.
\newblock \showarticletitle{Online payments by merely broadcasting messages}. In \bibinfo{booktitle}{\emph{DSN}}.
\newblock


\bibitem[{Committee on Payment and Settlement Systems}(2003)]%
        {BISfinality}
\bibfield{author}{\bibinfo{person}{{Committee on Payment and Settlement Systems}}.} \bibinfo{year}{2003}\natexlab{}.
\newblock \bibinfo{title}{A glossary of terms used in payments and settlement systems}.
\newblock \bibinfo{howpublished}{Bank for International Settlement (BIS) Report}.
\newblock


\bibitem[Cristian et~al\mbox{.}(1995)]%
        {cristian1995atomic}
\bibfield{author}{\bibinfo{person}{Flaviu Cristian}, \bibinfo{person}{Houtan Aghili}, \bibinfo{person}{Ray Strong}, {and} \bibinfo{person}{Danny Dolev}.} \bibinfo{year}{1995}\natexlab{}.
\newblock \showarticletitle{Atomic broadcast: From simple message diffusion to Byzantine agreement}.
\newblock \bibinfo{journal}{\emph{Information and Computation}} (\bibinfo{year}{1995}).
\newblock


\bibitem[Danezis et~al\mbox{.}(2022)]%
        {narwhal}
\bibfield{author}{\bibinfo{person}{George Danezis}, \bibinfo{person}{Lefteris Kokoris-Kogias}, \bibinfo{person}{Alberto Sonnino}, {and} \bibinfo{person}{Alexander Spiegelman}.} \bibinfo{year}{2022}\natexlab{}.
\newblock \showarticletitle{Narwhal and tusk: a {DAG}-based mempool and efficient {BFT} consensus}. In \bibinfo{booktitle}{\emph{EuroSys}}.
\newblock


\bibitem[Dwork et~al\mbox{.}(1988)]%
        {dwork1988consensus}
\bibfield{author}{\bibinfo{person}{Cynthia Dwork}, \bibinfo{person}{Nancy Lynch}, {and} \bibinfo{person}{Larry Stockmeyer}.} \bibinfo{year}{1988}\natexlab{}.
\newblock \showarticletitle{Consensus in the presence of partial synchrony}.
\newblock \bibinfo{journal}{\emph{Journal of the ACM (JACM)}} (\bibinfo{year}{1988}).
\newblock


\bibitem[Gelashvili et~al\mbox{.}(2022)]%
        {gelashvili2022jolteon}
\bibfield{author}{\bibinfo{person}{Rati Gelashvili}, \bibinfo{person}{Lefteris Kokoris-Kogias}, \bibinfo{person}{Alberto Sonnino}, \bibinfo{person}{Alexander Spiegelman}, {and} \bibinfo{person}{Zhuolun Xiang}.} \bibinfo{year}{2022}\natexlab{}.
\newblock \showarticletitle{Jolteon and ditto: Network-adaptive efficient consensus with asynchronous fallback}. In \bibinfo{booktitle}{\emph{Financial Cryptography and Data Security (FC)}}.
\newblock


\bibitem[Gelashvili et~al\mbox{.}(2023)]%
        {gelashvili2023block}
\bibfield{author}{\bibinfo{person}{Rati Gelashvili}, \bibinfo{person}{Alexander Spiegelman}, \bibinfo{person}{Zhuolun Xiang}, \bibinfo{person}{George Danezis}, \bibinfo{person}{Zekun Li}, \bibinfo{person}{Dahlia Malkhi}, \bibinfo{person}{Yu Xia}, {and} \bibinfo{person}{Runtian Zhou}.} \bibinfo{year}{2023}\natexlab{}.
\newblock \showarticletitle{Block-stm: Scaling blockchain execution by turning ordering curse to a performance blessing}. In \bibinfo{booktitle}{\emph{Principles and Practice of Parallel Programming}}.
\newblock


\bibitem[Guerraoui et~al\mbox{.}(2018)]%
        {DBLP:journals/corr/abs-1812-10844}
\bibfield{author}{\bibinfo{person}{Rachid Guerraoui}, \bibinfo{person}{Petr Kuznetsov}, \bibinfo{person}{Matteo Monti}, \bibinfo{person}{Matej Pavlovic}, {and} \bibinfo{person}{Dragos{-}Adrian Seredinschi}.} \bibinfo{year}{2018}\natexlab{}.
\newblock \showarticletitle{{AT2:} Asynchronous Trustworthy Transfers}. In \bibinfo{booktitle}{\emph{CoRR}}.
\newblock


\bibitem[Guerraoui et~al\mbox{.}(2019)]%
        {GuerraouiKMPS19}
\bibfield{author}{\bibinfo{person}{Rachid Guerraoui}, \bibinfo{person}{Petr Kuznetsov}, \bibinfo{person}{Matteo Monti}, \bibinfo{person}{Matej Pavlovic}, {and} \bibinfo{person}{Dragos{-}Adrian Seredinschi}.} \bibinfo{year}{2019}\natexlab{}.
\newblock \showarticletitle{The Consensus Number of a Cryptocurrency}. In \bibinfo{booktitle}{\emph{Principles of Distributed Computing (PODC)}}.
\newblock


\bibitem[Kokoris-Kogias et~al\mbox{.}(2018)]%
        {kokoris2018omniledger}
\bibfield{author}{\bibinfo{person}{Eleftherios Kokoris-Kogias}, \bibinfo{person}{Philipp Jovanovic}, \bibinfo{person}{Linus Gasser}, \bibinfo{person}{Nicolas Gailly}, \bibinfo{person}{Ewa Syta}, {and} \bibinfo{person}{Bryan Ford}.} \bibinfo{year}{2018}\natexlab{}.
\newblock \showarticletitle{Omniledger: A secure, scale-out, decentralized ledger via sharding}. In \bibinfo{booktitle}{\emph{Security and Privacy (SP)}}.
\newblock


\bibitem[Lamport(2019)]%
        {lamport2019time}
\bibfield{author}{\bibinfo{person}{Leslie Lamport}.} \bibinfo{year}{2019}\natexlab{}.
\newblock \showarticletitle{Time, clocks, and the ordering of events in a distributed system}.
\newblock In \bibinfo{booktitle}{\emph{Concurrency: the Works of Leslie Lamport}}.
\newblock


\bibitem[Malkhi and Nayak(2023)]%
        {malkhi2023hotstuff}
\bibfield{author}{\bibinfo{person}{Dahlia Malkhi} {and} \bibinfo{person}{Kartik Nayak}.} \bibinfo{year}{2023}\natexlab{}.
\newblock \showarticletitle{HotStuff-2: Optimal Two-Phase Responsive BFT}.
\newblock \bibinfo{journal}{\emph{Cryptology ePrint}} (\bibinfo{year}{2023}).
\newblock


\bibitem[Meta(2021a)]%
        {fastpay-code}
\bibfield{author}{\bibinfo{person}{Meta}.} \bibinfo{year}{2021}\natexlab{a}.
\newblock \bibinfo{title}{{FastPay: High-Performance Byzantine Fault Tolerant Settlement}}.
\newblock \bibinfo{howpublished}{\url{https://github.com/novifinancial/fastpay}}.
\newblock


\bibitem[Meta(2021b)]%
        {narwhal-code}
\bibfield{author}{\bibinfo{person}{Meta}.} \bibinfo{year}{2021}\natexlab{b}.
\newblock \bibinfo{title}{{Narwhal and Tusk: A DAG-based Mempool and Efficient {BFT} Consensus.}}
\newblock \bibinfo{howpublished}{\url{https://github.com/facebookresearch/narwhal}}.
\newblock


\bibitem[Meta(2023)]%
        {rocksdb-code}
\bibfield{author}{\bibinfo{person}{Meta}.} \bibinfo{year}{2023}\natexlab{}.
\newblock \bibinfo{title}{{A persistent key-value store for fast storage environments}}.
\newblock \bibinfo{howpublished}{\url{https://rocksdb.org}}.
\newblock


\bibitem[Nakamoto(2008)]%
        {nakamoto2008bitcoin}
\bibfield{author}{\bibinfo{person}{Satoshi Nakamoto}.} \bibinfo{year}{2008}\natexlab{}.
\newblock \showarticletitle{Bitcoin: A peer-to-peer electronic cash system}.
\newblock \bibinfo{journal}{\emph{Decentralized business review}} (\bibinfo{year}{2008}).
\newblock


\bibitem[Quinn(2023)]%
        {quinn-code}
\bibfield{author}{\bibinfo{person}{Quinn}.} \bibinfo{year}{2023}\natexlab{}.
\newblock \bibinfo{title}{{Async-friendly QUIC implementation in Rust}}.
\newblock \bibinfo{howpublished}{\url{https://github.com/quinn-rs/quinn}}.
\newblock


\bibitem[Saraph and Herlihy(2019)]%
        {saraph2019empirical}
\bibfield{author}{\bibinfo{person}{Vikram Saraph} {and} \bibinfo{person}{Maurice Herlihy}.} \bibinfo{year}{2019}\natexlab{}.
\newblock \showarticletitle{An empirical study of speculative concurrency in ethereum smart contracts}.
\newblock \bibinfo{journal}{\emph{Arxiv}} (\bibinfo{year}{2019}).
\newblock


\bibitem[Schneider(1990)]%
        {schneider1990implementing}
\bibfield{author}{\bibinfo{person}{Fred~B Schneider}.} \bibinfo{year}{1990}\natexlab{}.
\newblock \showarticletitle{Implementing fault-tolerant services using the state machine approach: A tutorial}.
\newblock \bibinfo{journal}{\emph{Computing Surveys (CSUR)}} (\bibinfo{year}{1990}).
\newblock


\bibitem[Sliwinski et~al\mbox{.}(2022)]%
        {sliwinski2022consensus}
\bibfield{author}{\bibinfo{person}{Jakub Sliwinski}, \bibinfo{person}{Yann Vonlanthen}, {and} \bibinfo{person}{Roger Wattenhofer}.} \bibinfo{year}{2022}\natexlab{}.
\newblock \showarticletitle{Consensus on demand}. In \bibinfo{booktitle}{\emph{Stabilization, Safety, and Security of Distributed Systems}}.
\newblock


\bibitem[Sliwinski and Wattenhofer(2019)]%
        {abc}
\bibfield{author}{\bibinfo{person}{Jakub Sliwinski} {and} \bibinfo{person}{Roger Wattenhofer}.} \bibinfo{year}{2019}\natexlab{}.
\newblock \bibinfo{title}{ABC: Asynchronous Blockchain without Consensus}.
\newblock \bibinfo{howpublished}{Arxiv}.
\newblock


\bibitem[Sonnino(2023)]%
        {bullshark-code}
\bibfield{author}{\bibinfo{person}{Alberto Sonnino}.} \bibinfo{year}{2023}\natexlab{}.
\newblock \bibinfo{title}{{Implementation of BFT consensus protocols based on the Narwhal mempool.}}
\newblock \bibinfo{howpublished}{\url{https://github.com/asonnino/narwhal/tree/bullshark}}.
\newblock


\bibitem[Sonnino et~al\mbox{.}(2020)]%
        {sonnino2020replay}
\bibfield{author}{\bibinfo{person}{Alberto Sonnino}, \bibinfo{person}{Shehar Bano}, \bibinfo{person}{Mustafa Al-Bassam}, {and} \bibinfo{person}{George Danezis}.} \bibinfo{year}{2020}\natexlab{}.
\newblock \showarticletitle{Replay attacks and defenses against cross-shard consensus in sharded distributed ledgers}. In \bibinfo{booktitle}{\emph{European Symposium on Security and Privacy (EuroS\&P)}}.
\newblock


\bibitem[Spiegelman et~al\mbox{.}(2022a)]%
        {bullshark}
\bibfield{author}{\bibinfo{person}{Alexander Spiegelman}, \bibinfo{person}{Neil Giridharan}, \bibinfo{person}{Alberto Sonnino}, {and} \bibinfo{person}{Lefteris Kokoris-Kogias}.} \bibinfo{year}{2022}\natexlab{a}.
\newblock \showarticletitle{Bullshark: Dag {BFT} protocols made practical}. In \bibinfo{booktitle}{\emph{CCS}}.
\newblock


\bibitem[Spiegelman et~al\mbox{.}(2022b)]%
        {spiegelman2022bullshark}
\bibfield{author}{\bibinfo{person}{Alexander Spiegelman}, \bibinfo{person}{Neil Giridharan}, \bibinfo{person}{Alberto Sonnino}, {and} \bibinfo{person}{Lefteris Kokoris-Kogias}.} \bibinfo{year}{2022}\natexlab{b}.
\newblock \showarticletitle{Bullshark: The Partially Synchronous Version}.
\newblock \bibinfo{journal}{\emph{Arxiv}} (\bibinfo{year}{2022}).
\newblock


\bibitem[Stathakopoulou et~al\mbox{.}(2022)]%
        {stathakopoulou2022state}
\bibfield{author}{\bibinfo{person}{Chrysoula Stathakopoulou}, \bibinfo{person}{Matej Pavlovic}, {and} \bibinfo{person}{Marko Vukoli{\'c}}.} \bibinfo{year}{2022}\natexlab{}.
\newblock \showarticletitle{State machine replication scalability made simple}. In \bibinfo{booktitle}{\emph{EuroSys}}.
\newblock


\bibitem[Sui(2024a)]%
        {sui}
\bibfield{author}{\bibinfo{person}{Sui}.} \bibinfo{year}{2024}\natexlab{a}.
\newblock \bibinfo{title}{Build Beyond}.
\newblock \bibinfo{howpublished}{\url{https://sui.io}}.
\newblock


\bibitem[Sui(2024b)]%
        {sui-explorer}
\bibfield{author}{\bibinfo{person}{Sui}.} \bibinfo{year}{2024}\natexlab{b}.
\newblock \bibinfo{title}{Explore Sui Blockchain}.
\newblock \bibinfo{howpublished}{\url{https://suivision.xyz}}.
\newblock


\bibitem[team(2023)]%
        {tokio-code}
\bibfield{author}{\bibinfo{person}{The~Tokio team}.} \bibinfo{year}{2023}\natexlab{}.
\newblock \bibinfo{title}{{Build reliable network applications without compromising speed.}}
\newblock \bibinfo{howpublished}{\url{https://tokio.rs}}.
\newblock


\bibitem[Zhang and Anand(2022)]%
        {zhang2022ethereum}
\bibfield{author}{\bibinfo{person}{Weijia Zhang} {and} \bibinfo{person}{Tej Anand}.} \bibinfo{year}{2022}\natexlab{}.
\newblock \showarticletitle{Ethereum Architecture and Overview}.
\newblock In \bibinfo{booktitle}{\emph{Blockchain and Ethereum Smart Contract Solution Development: Dapp Programming with Solidity}}. \bibinfo{publisher}{Springer}.
\newblock


\end{thebibliography}

\appendix
\section{Detailed \sysname System} \label{sec:detailed-design}
We present the \sysname core protocol by providing algorithms and specifying the checks performed by validators at each step of the protocol.

\subsection{Objects Operations} \label{sec:objects-operations}
\sysname validators rely on a set of objects to represent the current and historical state of the replicated system. An \emph{Object} ($\object$) stores user smart contracts and data within \sysname. Transactions can affect objects which can be $\created$, $\mutated$, $\wrapped$, $\unwrapped$ and $\deleted$. 

Calling $\okey{\object}$ returns the \emph{key} (\objectkey) of the object, namely a tuple (\objectid, \objectversion). \objectid is cryptographically derived so that finding collisions is infeasible. Versions monotonically increase with each transaction processing the object and are determined via Lamport timestamps~\cite{lamport2019time}. Calls to $\version{\object}$ and $\initialversion{\object}$ return the current and initial version of the object respectively.

The $\oauth{\object}$ is either the owner's public key that may use this object, or the $\objectid$ of another \emph{parent} object, in which case this is a \emph{child} object. $\owned{\object}$ returns whether the object has an owner, or whether it is read-only or a shared object (see below).  The last transaction digest (\transactiondigest) that last mutated or created the object $\parent{\object}$.

\subsection{Protocol Messages} \label{sec:messages}
Validators and users run the core protocol described in \Cref{sec:core-protocol} by exchanging the following messages.

\noindent \textbf{Transactions.}
A \emph{transaction} ($\transaction$) is a structure representing a state transition for one or more objects. They support a few self-explanatory access operations, such as to get its digest $\txdigest{\transaction}$, and different types of input objects $\txinputs{\transaction}$ (object reference), $\txreadonlyinputs{\transaction}$, $\txsharedinputs{\transaction}$ (object ID and initial version), and $\txecon{\transaction}$ (the reference to the gas object to pay fees).

A Transaction may be checked for validity given a set of input objects, or can be executed to compute output objects:
\begin{itemize}[leftmargin=*]
    \item $\txvalid{\transaction, [\object]}$ returns true if the transaction is valid, given the requested input objects provided. This check verifies that the transactions are authorized to act on the input objects, as well as sufficient gas is available. Transaction validity, as returned by $\txvalid{\transaction, [\object]}$ can be determined statically without executing the Move contract.
    \item $\txexec{\transaction, [\object_{o}], [\object_{s}]}$ executes using the
          MoveVM~\cite{move_white} and returns a structure $\effects$ representing its effects along with the output objects $[\object_{out}]$. The output objects are the new objects $\created$, $\mutated$, $\wrapped$, $\unwrapped$ and $\deleted$ by the transaction. When objects are created, updated, or unwrapped their version number is the Lamport timestamp of the transaction.
          $[\object_{o}]$ and $[\object_{s}]$ respectively represent the owned and shared input objects. A valid transaction execution is infallible and has deterministic output.
\end{itemize}%
A transaction is indexed by the $\transactiondigest$ over its raw data
, which also authenticates its full contents. All valid transactions
have at least one owned input, namely the objects used to pay for gas. We note that a \transaction, can contain multiple related commands at once that execute atomically, and we call these \emph{Programmable Transaction Blocks} (PTBs) and show they increase the throuhgput of the system significantly.

\noindent \textbf{Certificates.}
A \emph{transaction certificate} ($\cert$) contains the transaction itself, the signature of the user authorizing the use of the input objects, and the identifiers and signatures from a quorum of $2f+1$ validators or more. For simplicity, we assume that every operation defined over a transaction is also defined over a certificate. For instance, ``$\txdigest{\transaction}$'' is equivalent to ``$\txdigest{\cert}$''.
A certificate may not be unique, and the same logical certificate may be signed by a different quorum of validators or even have different authorization paths (e.g., a 2-out-of-3 multisig). However, two different valid certificates on the same transaction should be treated as representing semantically the same certificate. The identifiers of signers are included in the certificate to identify validators ready to process the certificate, or that may be used to download past information required to process the certificate.
The signatures are aggregated (e.g., using BLS~\cite{BonehLS01}), compressing the quorum of signers to a single signature.


\noindent \textbf{Transaction effects.}
A \emph{transaction effects} (\effects) structure summarizes the outcome of a transaction execution. Its digest and authenticator is computed by $\txdigest{\effects}$. It supports operations to access its data such as $\efftransaction{\effects}$ (returns the transaction digest) and $\txdeps{\effects}$ (the digest of all transactions to create input objects). The $\contents{\effects}$ returns a summary of the execution: $\status$ reports the outcome of the smart contract execution. The lists $\created$, $\mutated$, $\wrapped$, $\unwrapped$, and $\deleted$, list the object references that underwent the respective operations. Finally, $\events$ lists the events emitted by the execution.

\noindent \textbf{Partial certificate.}
A \emph{partial certificate} (\txsign) contains the same information, but signatures from a set of validators representing stake lower than the required quorum, usually a single one. We call \emph{signed transaction} a partial certificate signed by one validator.

\noindent \textbf{Effect certificates.}
Similarly, an \emph{effects certificate} ($\ecert$) on an effects structure contains the effects structure itself, and signatures from validators
that represent a quorum for the epoch in which the transaction is valid. The same caveats, about non-uniqueness and identity apply as for transaction certificates. A partial effects certificate, usually containing a single validator signature and the effects structure is denoted as $\esign$. For transactions that only include owned objects, this certificate provides both finality and settlement. For transactions with shared objects validators can first reply with an empty effects structure (so when an empty effects certificate is formed the system reaches finality) and then add the effects structure and re-sign after consensus (to reach settlement). For the rest of the paper, we omit the transmission of the empty effects structure for simplicity.

\subsection{Data Structures} \label{sec:datastructures}
Validators maintain a set of persistent tables abstracted as key-value maps, with the usual $\mathsf{contains}$, $\mathsf{get}$, and $\mathsf{set}$ operations.

Reliable broadcast on owned objects uses the \emph{owned lock map} ($\ownedlockdb[\objectkey] \rightarrow \txsign \textsf{Option}$) which records the first valid transaction $\transaction$ seen and signed by the validator for an owned object's $\objectkey$, or None if the $\objectkey$ exists but no valid transaction using it as an input has been seen.
The \emph{certificate map} ($\certdb[\transactiondigest] \rightarrow (\cert, \esign)$) records all full certificates $\cert$, including $\transaction$, processed by the validator, along with their signed effects $\esign$.

To manage the execution of shared object transactions, the \emph{shared lock map} ($\sharedlockdb[(\transactiondigest, \ObjID)] \rightarrow \Version$) records the version number of $\ObjID$ assigned to a transaction $\transactiondigest$.
The \emph{next shared lock map} ($\nextsharedlockdb[\ObjID] \rightarrow \Version$) records the next available version (we discuss their use in \Cref{sec:operations}).

The \emph{object map} ($\objdb[\objectkey] \rightarrow \object$) records all objects $\object$ created by executed certificates within $\certdb$ by object key, and also allows lookups for the latest known object version.
This store can be completely derived by re-executing all certificates in $\certdb$. Only the latest version is necessary to process new transactions and older versions are only maintained to facilitate parallel execution, reconfiguration, reads, and audits.




Only ($\ownedlockdb, \sharedlockdb, \nextsharedlockdb$) require key self-consistency, namely a read on a key should always return whether a value or None is present for a key that exists, and such a check should be atomic with an update that sets a lock to a non-None value. This is a weaker property than strong consistency across keys and allows for efficient sharding of the store. To ensure this property, they are only updated by a single task (see \Cref{sec:operations}). The other stores may be eventually consistent without affecting safety.

\subsection{Validator Core Operation} \label{sec:operations}
Validators process transactions and certificates as described in \Cref{sec:core-protocol}. Transactions are submitted to the validator core for processing by users, while certificates can either be submitted by users or by the consensus engine.

\noindent \textbf{Process Transaction.}
\Cref{alg:process-tx} shows how \sysname processes transactions; that is, step~\two of \Cref{fig:sui-overview} (see \Cref{sec:core-protocol}). The function \textsf{LoadObjects} (\Cref{alg:line:process-tx-load-objects}) simply loads the specified object(s) from the $\objdb$ store; \textsf{LoadLatestVersionObjects} (\Cref{alg:line:process-tx-load-latest-objects}) loads the latest version of the specified object(s) from the $\objdb$ store; and \textsf{AcquireLocks} (\Cref{alg:line:acquire_tx_locks}) acquires a mutex for every owned-object transaction input.
Upon receiving a transaction $\transaction$ a validator calls \textsf{ProcessTx} to perform a number of checks:
\begin{enumerate}
    \item It ensures all object references $\txinputs{\transaction}$ and the gas object reference in $\txecon{\transaction}$ exist in the $\objdb$ store and loads them into $[\object]$. For owned objects both the id and version should be available; for read-only or shared objects the object ID should exist, and for shared objects, the initial version specified in the input must match the initial version of the shared object returned by $\initialversion{\cdot}$.  This check implicitly ensures that all owned objects have the version number specified in the transaction since the call to \textsf{LoadObjects} (\Cref{alg:line:process-tx-load-objects}) loads the pair $\objectkey = (\objectid, \Version)$ from the $\objdb$ store; $\objdb$ store holds a single entry (the latest version) per object.
    \item It checks $\txvalid{\transaction, [\object]}$ is true. This step ensures the authentication information in the transaction allows access to the owned objects. That is, (i) the signer of the transaction must be the owner of all the input objects owned by an address input; (ii) the parent object of any included child object owned by another object should be an input to the transaction; and (iii) sufficient gas can be made available to cover the minimum cost of execution.
    \item It ensures it can acquire a lock for every owned-object transaction input; otherwise, it returns an error. Acquiring a lock ensures that no other task can concurrently perform the next step of the algorithm on the same input objects.
    \item It checks that $\ownedlockdb[\objectkey]$ for all owned $\txinputs{\transaction}$ objects exist, and it is either None or set to \emph{the same} $\transaction$, and atomically sets it to $\txsign$. In other words, for each owned input version in the transactions: (i) a key for this object \emph{exists} in $\ownedlockdb$ and (ii) no other transaction $\transaction' \neq \transaction$ has been assigned as a value for this object version in $\ownedlockdb$, i.e.\ $\transaction$ is the first valid transaction seen using this input object. This is a key validity check to implement \emph{Byzantine consistent broadcast}~\cite{cachinBook} and ensure safety.
\end{enumerate}

Transaction processing ends if any of the checks fail and an error is returned.
If all checks are successful then the validator returns a signature on the transaction, ie.\ a partial certificate \txsign. Processing a transaction is idempotent upon success, and always returns a partial certificate (\txsign) within the same epoch. Any party may collate a transaction and signatures (\txsign) from a set of validators forming a quorum for epoch $e$, to form a transaction certificate $\cert$. It is \emph{critical that this step happens for all transactions as it acts as spam protection} for order-execute consensus engines. In the original Bullshark~\cite{bullshark} work it is trivial to drop the throughput to zero by a single malicious client sending corrupted transactions or duplicates. \sysname prevents such attacks thanks to the stronger coupling with the state that allows only transactions with gas to make it to consensus.

Many tasks can call \textsf{ProcessTx} concurrently. \sysname only acquires mutexes on the minimum amount of data: the owned-objects transaction inputs (\Cref{alg:process-tx} \Cref{alg:line:acquire_tx_locks}).

\begin{algorithm}[t]
    \caption{Process transaction}
    \label{alg:process-tx}
    \footnotesize
    \begin{algorithmic}[1]
        \Statex // Executed upon receiving a transaction from a user.
        \Statex // Many tasks can call this function.
        \Procedure{ProcessTx}{$\transaction$}
        \State // Check 1.1: Ensure all objects exist.
        \State $[\object_o] = \Call{LoadObjects}{\txinputs{\transaction}}$ \label{alg:line:process-tx-load-objects}
        \State $[\object_r] = \Call{LoadLatestVersionObjects}{\txreadonlyinputs{\transaction}}$ \label{alg:line:process-tx-load-latest-objects}
        \For{$(\objectid, \InitVersion) \in \txsharedinputs{\transaction}$}
        \State $\object_s = \Call{LoadLatestVersionObjects}{\objectid}$
        \If{$\initialversion{\object_s} \neq \InitVersion$}
        \State \Return Error
        \EndIf
        \EndFor
        \State
        \State // Check 1.2: Check the transaction's validity (see \Cref{sec:messages}).
        \If{!$\txvalid{\transaction, [\object_o]}$} \Return Error \EndIf
        \State
        \State // Check 1.3: Try to acquire a mutex over $\txinputs{\transaction}$
        \State guard = \Call{AcquireLocks}{$\transaction$} \Comment{Error if cannot acquire all locks} \label{alg:line:acquire_tx_locks}
        \State
        \State // Check 1.4: Lock all owned-objects.
        \State $\txsign = \sign{\transaction}$
        \For{$\objectkey \in \txinputs{\transaction}$}
        \If{$\ownedlockdb[\objectkey] == \textsf{None}$}
        \State $\ownedlockdb[\objectkey] = \txsign$ \label{alg:line:assign_owned_lock}
        \ElsIf{$\ownedlockdb[\objectkey] \neq \txsign$} \label{alg:line:no-conflict}
        \State \Return Error
        \EndIf
        \EndFor
        \State
        \State \Return \txsign
        \EndProcedure
    \end{algorithmic}
\end{algorithm}

\begin{algorithm}[t]
    \caption{Storage support (generic)}
    \label{alg:storage-support}
    \footnotesize
    \begin{algorithmic}[1]
        \Statex // Check whether the certificate as already been executed.
        \Procedure{AlreadyExecuted}{\cert}
        \State $\transactiondigest = \txdigest{\cert}$
        \State $(\_, \esign) = \certdb[\transactiondigest]$
        \If{$\esign$}
        \Return $\esign$
        \EndIf
        \Return \textbf{None}
        \EndProcedure

        \Statex
        \Statex // All operations inside this function are atomic.
        \Statex // Note that all owned and shared objects have the same version $v$.
        \Procedure{AtomicPersist}{$\cert, \esign, [\Obj_{out}]$}
        \State $\transactiondigest = \txdigest{\cert}$
        \State $\certdb[\transactiondigest] = (\cert, \esign)$
        \For{$\Obj \in \Obj_{out}$}
        \State $\objectkey = \okey{\Obj}$
        \State $\objdb[\objectkey] = \Obj$
        \If{$\owned{\Obj}$}
        \State $\ownedlockdb[\objectkey]$ = None
        \EndIf
        \EndFor
        \EndProcedure
    \end{algorithmic}
\end{algorithm}

\begin{algorithm}[t]
    \caption{Storage support (shared objects)}
    \label{alg:storage-support-shared-objects}
    \footnotesize
    \begin{algorithmic}[1]
        \Statex // Assign locks to the shared objects referenced by \cert.
        \Procedure{WriteSharedLocks}{$\cert, v$}
        \State $\transactiondigest = \txdigest{\cert}$
        \For{$(\ObjID, \InitVersion) \in \txsharedinputs{\cert}$}
        \State $\Version = \nextsharedlockdb[\ObjID] \; || \; \InitVersion$
        \State $\sharedlockdb[(\transactiondigest, \ObjID)] = \Version$
        \State $\nextsharedlockdb[\ObjID] = v + 1$ \Comment{Lamport timestamp}
        \EndFor
        \EndProcedure

        \Statex
        \Statex // Check whether all shared objects referenced by \cert are locked.
        \Procedure{SharedLocksExist}{$\cert$}
        \State $\transactiondigest = \txdigest{\cert}$
        \For{$\ObjID \in \txsharedinputs{\cert}$}
        \If{$!\sharedlockdb[(\transactiondigest, \ObjID)]$}
        \State \Return false \Comment{lock not found}
        \EndIf
        \EndFor
        \State \Return true
        \EndProcedure

        \Statex
        \Statex // Ensure that $\cert$ is the next certificate scheduled for execution.
        \Procedure{CheckSharedLocks}{$\cert$}
        \State $\transactiondigest = \txdigest{\cert}$
        \For{$\ObjID \in \txsharedinputs{\cert}$}
        \State $\Obj = \objdb[\ObjID]$
        \State $\Version = \version{\Obj}$
        \If{$\sharedlockdb[(\transactiondigest, \ObjID)] \neq \Version$}
        \State \Return false
        \EndIf
        \EndFor
        \State \Return true
        \EndProcedure
    \end{algorithmic}
\end{algorithm}

\begin{algorithm}[t]
    \caption{Process certificate}
    \label{alg:process-cert}
    \footnotesize
    \begin{algorithmic}[1]
        \Require Input certificate (\cert) is signed by a quorum

        \Statex
        \Statex // Executed upon receiving a certificate from a user.
        \Statex // Many tasks can call this function.
        \Procedure{ProcessCert}{$\cert$}
        \State // Check 4.1: Ensure the $\cert$ is for the current epoch \epoch.
        \If{$\txepoch{\cert} \neq \epoch$} \Return Error \EndIf
        \State
        \State // Check 4.2: Load objects from store, return error if missing object.
        \State [$\object_o$] = \Call{LoadObjects}{$\txinputs{\cert}$} \label{alg:line:process-cert-load-owned-objects}
        \State [$\object_{s}$] = \Call{LoadObjects}{$\txsharedinputs{\cert}$} \label{alg:line:process-cert-load-shared-objects}
        \State
        \State // Check 4.3: Check the objects locks.
        \If{!\Call{SharedLocksExist}{$\cert$}} \label{alg:line:ensure-shared-locks-exist}
        \State // Sequence the certificates
        \State \Call{ForwardToConsensus}{$\cert$} \label{alg:line:forward-to-consensus}
        \State \Return
        \EndIf
        \If{!\Call{CheckSharedLocks}{$\cert$}} \Return Error \EndIf \label{alg:line:ensure-shared-locks-correct}
        \State
        \State // Execute the certificate. \label{alg:line:execute-cert}
        \State $(\esign, [\Obj_{out}]) = \txexec{\cert, [\object_{o}], [\object_{s}]}$ \label{alg:line:exec}
        \State \Call{AtomicPersist}{$\cert, \esign, [\Obj_{out}]$} \label{alg:line:atomic-persist}
        \State \Return \esign
        \EndProcedure

        \Statex
        \Statex // Executed upon receiving a certificate from consensus.
        \Statex // This function must be called by a single task.
        \Procedure{AssignSharedLocks}{$\cert$}
        \State // Ensure shared locks are assigned only once.
        \If{\Call{SharedLocksExist}{$\cert$}}
        \State \Return
        \EndIf
        \State
        \State // Extract the highest objects version.
        \State $v_o = 1$
        \For{$(\objectid, \Version) \in \txinputs{\cert}$}
        \State $v_o = \max(v_o, \Version)$
        \EndFor
        \State $v_s = 1$
        \For{$(\objectid, \InitVersion) \in \txsharedinputs{\cert}$}
        \State $\Version = \nextsharedlockdb[\objectid] \; || \; \InitVersion$
        \State $v_s = \max(vs, \Version)$
        \EndFor
        \State $v_{max} = \max(v_o, v_s)$ \Comment{Lamport timestamp} \label{alg:line:lamport-timestamp}
        \State
        \State // Lock all shared objects to $v_{max}$.
        \State \Call{WriteSharedLocks}{$\cert, v_{max}$}  \label{alg:line:shared-lock-write}
        \EndProcedure
    \end{algorithmic}
\end{algorithm}

\para{Process user certificates}
\Cref{alg:process-cert} shows how \sysname processes certificates; that is, step~\five of \Cref{fig:sui-overview} (see \Cref{sec:core-protocol}). Algorithms \ref{alg:storage-support} and \ref{alg:storage-support-shared-objects} provide non-trivial support functions. Every certificate input to a function of \Cref{alg:process-cert} is first checked to ensure it is signed by $2f+1$ validators.
Upon receiving a certificate $\cert$ a validator calls \textsf{ProcessCert} to perform a number of checks:
\begin{enumerate}[leftmargin=*]
    \item Ensure $\txepoch{\cert}$ is the current epoch. This is a property of the quorum of signatures forming the certificate.
    \item Load the owned objects (\Cref{alg:line:process-cert-load-owned-objects}) and shared objects (\Cref{alg:line:process-cert-load-shared-objects}). Success in loading these objects ensures the validator already processed all past certificates concerning the loaded objects. If any object is missing, the validator aborts and returns an error. When loading shared objects, the function \textsf{LoadObjects} returns the shared object with the highest version number.
    \item If the certificate contains shared objects, it ensures the certificate is already sequenced by the consensus engine; otherwise, it forwards the certificate to consensus (\Cref{alg:line:forward-to-consensus}). Finally, it checks the shared locks to ensure the current certificate $\cert$ is the next scheduled certificate.
\end{enumerate}

If all check succeeds, the transaction can be executed. The validator then atomically persists the execution results to storage (function \textsf{AtomicPersist} of \Cref{alg:storage-support}).
It inserts the new (or mutated) objects in the $\objdb$ store and sets the version number of all owned and shared objects to the highest version number amongst all objects in the transaction plus 1 (Lamport timestamp); that is, the new version is
$v = 1 + \max_{o \in [\object_{o}] \cup [\object_{s}]} \version{o}$.
The validator also persists the certificate along with the effects resulting from its execution;
and updates the owned-object lock store to unblock future transactions using these objects. If transaction execution fails, \sysname unlocks any owned objects used as input of the transaction. That is, it sets $\ownedlockdb[\objectkey] = \text{None}, \;\forall\objectkey \in \txinputs{\transaction}$. Unlocking owned objects is essential to allow future transactions to re-use them. Gas payment is deduced from the payment objects whether the execution succeeds or fails.

\para{Process consensus certificates}
Upon receiving a certificate output from consensus, the validator calls \textsc{AssignSharedLocks} (\Cref{alg:process-cert}) to lock the transaction's shared-objects to a version number. This can be done without executing the transactions, only by inspecting and updating the transaction as well as the $\sharedlockdb$ and $\nextsharedlockdb$ tables. When an entry for a shared object does not exist in the tables it is assigned the initial version number given in the transaction input. 
Otherwise, it is given the value in the $\nextsharedlockdb$ table. The $\nextsharedlockdb$ is updated with the Lamport timestamp~\cite{lamport2019time} of the transaction: the highest version of all input objects used plus one. \textsc{AssignSharedLocks} must be only called by a single task.
After successfully calling \textsc{AssignSharedLocks}, the validator can call (again) \textsf{ProcessCert} to execute the certificate.

Shared objects may be included in a transaction explicitly only for reads (we omit this special case from the algorithms for clarity). In that case, the transaction is assigned in $\sharedlockdb$ the version in the $\nextsharedlockdb$ table for the shared object. However, the version of the object in $\nextsharedlockdb$ is not increased, and upon transaction execution, the shared object is not mutated. In order to preserve the safety of dynamic accesses we make sure that within a Bullshark commit (level of concurrency) all read-only transactions on shared objects are executed on the initial version V and only after writes are allowed to execute and mutate the object. This facility allows multiple transactions executing in parallel to use the same shared object for reads. For example, it is used to update a clock object with the current system time upon each commit that transactions may read concurrently.

\noindent \textbf{Additional checks.}
Algorithms~\ref{alg:process-tx} and~\ref{alg:process-cert} only describe the core validator operations. In practice, validators perform a number of extra checks to early reject duplicate messages. For instance, validators can easily check whether a certificate has already been executed by calling \textsf{AlreadyExecuted} (\Cref{alg:storage-support}). Such a check is useful to improve performance and prevent obvious DoS attacks but is not strictly needed for security (it does not guarantee idempotent validator replies on its own since both \textsf{ProcessTx} and \textsf{ProcessCert} can be called concurrently by multiple tasks).

\section{Benchmark Setup}\label{app:Benchmark}
We deploy a fully-featured \sysname testbed on AWS, using \texttt{m5d.8xlarge} instances across 13 different AWS regions: N. Virginia (us-east-1), Oregon (us-west-2), Canada (ca-central-1), Frankfurt (eu-central-1), Ireland (eu-west-1), London (eu-west-2), Paris (eu-west-3), Stockholm (eu-north-1), Mumbai (ap-south-1), Singapore (ap-southeast-1), Sydney (ap-southeast-2), Tokyo (ap-northeast-1), and Seoul (ap-northeast-2).
Validators are distributed across those regions as equally as possible. Each machine provides 10Gbps of bandwidth, 32 virtual CPUs (16 physical core) on a 2.5GHz, Intel Xeon Platinum 8175, 128GB memory, and runs Linux Ubuntu server 22.04. \sysname persists all data on the NVMe drives provided by the machine (rather than the root partition). We select these machines because they provide decent performance and are in the price range of `commodity servers'.

In all graphs of the paper, each data point is the average of the latency of all transactions of the run, and the error bars represent one standard deviation (errors bars are sometimes too small to be visible on the graph). 
We instantiate several geo-distributed benchmark clients submitting transactions at a fixed rate for a duration of 10 minutes; unless specified otherwise each benchmark client submits at most 350 tx/s and the number of clients thus depends on the desired input load. 
\section{Reconfiguration Appendix} \label{sec:reconfig-appendix}
This appendix complements \Cref{sec:reconfiguration-protocol} by providing the detailed algorithm of the \sysname reconfiguration logic and by presenting benchmarks demonstrating that \sysname's performance is largely unaffected by epoch changes.

\subsection{Algorithm}
\Cref{alg:reconfiguration} shows the \sysname reconfiguration logic coded as a smart contract.

\begin{algorithm}[th!]
    \caption{Reconfiguration Contract}
    \label{alg:reconfiguration}
    \footnotesize
    \begin{algorithmic}[1]
        \Statex // Smart contract state
        \Statex \textsf{T} \Comment{Minimum stake to become a validator}
        \Statex \textsf{S} \Comment{Last sequence number before starting epoch change}
        \Statex $\textsf{total\_old\_stake} = \textsf{GenesisStake}$  \Comment{Total stake of the old committee}
        \Statex $\textsf{total\_new\_stake} = 0$  \Comment{Total stake of the new committee}
        \Statex $\textsf{old\_keys} = \textsf{GenesisKeys}$ \Comment{Identifiers of the old committee members}
        \Statex \textsf{new\_keys} = $\{\}$ \Comment{Identifiers of the new committee members}
        \Statex $\textsf{epoch\_edge} = 0$ \Comment{Sequence number of the last epoch's checkpoint}
        \Statex \textsf{state} = \textsf{Register} \Comment{Current state of the smart contract}
        \Statex \textsf{stake} = 0 \Comment{Variable counting the accumulated stake}

        \Statex
        \Statex // Step 1: Users register to become the next validators.
        \Function{Register}{$\texttt{sender}$}
        \If{$\textsf{state} \neq \textsf{Register}$} \Return \EndIf
        \If{$\texttt{sender.stake} \geq \textsf{T}$}
        \State $\textsf{new\_keys} = \textsf{new\_keys} \cup \texttt{sender}$
        \State \textsf{total\_new\_stake} += \texttt{sender.stake}
        \EndIf
        \EndFunction

        \Statex
        \Statex // Step 2: New validators signal they are ready to take over.
        \Function{Ready}{$\texttt{sender}$}
        \State $seq = \Call{GetLatestCheckpointSeq}$
        \If{$\textsf{state} = \textsf{Register}  \text{ and } seq \geq \textsf{S}$} \label{alg:line::ready}
        \State $\textsf{state} = \textsf{Ready}$
        \EndIf
        \If{$\textsf{state} \neq \textsf{Ready}$} \Return \EndIf
        \If{$\texttt{sender} \in \textsf{new\_keys}$}
        \State \textsf{stake} += \texttt{sender.stake}
        \EndIf
        \If{$\textsf{stake} \geq 2*\textsf{total\_new\_stake}/3+1$}
        \State $\textsf{state} = \textsf{End-of-Epoch}$
        \State $\textsf{stake} = 0$
        \State \Call{PauseTxLocking}{} \Comment{Stop signing messages} \label{alg:line:pause-tx-locking}
        \EndIf
        \EndFunction

        \Statex
        \Statex // Step 3: Old validators signal the epoch can safely finish.
        \Function{End-of-Epoch}{$\texttt{sender}$}
        \If{$\textsf{state} \neq \textsf{End-of-Epoch}$} \Return \EndIf
        \If{$\texttt{sender} \in \textsf{old\_keys}$}
        \State \textsf{stake} += \texttt{sender.stake}
        \EndIf
        \If{$\textsf{stake} \geq 2*\textsf{total\_old\_stake}/3+1$}
        \State $\textsf{state} = \textsf{Handover}$
        \State $\textsf{stake} = 0$
        \State $\textsf{epoch\_edge} = \Call{GetLatestCheckpointSeq}$ \label{alg:line:latest-checkpoint_seq}
        \EndIf
        \EndFunction

        \Statex
        \Statex // Step 4: The new validators take over.
        \Function{Handover}{$\texttt{sender}$}
        \If{$\textsf{state} \neq \textsf{Handover}$} \Return \EndIf
        \State $seq = \Call{GetLatestCheckpointSeq}$
        \If{$seq \geq \textsf{epoch\_edge} + 1$} \label{alg:line:one-more}
        \State $\textsf{old\_keys} = \textsf{new\_keys}$
        \State $\textsf{total\_old\_stake} = \textsf{total\_new\_stake}$
        \State $\textsf{state} = \textsf{Register}$
        \State $\textsf{new\_keys} = \{\}$
        \State $\textsf{total\_old\_stake} = 0$, $\textsf{epoch\_edge} = 0$
        \State \Call{Shutdown}{} \Comment{Old validators can shutdown} \label{alg:line:shutdown}
        \EndIf
        \EndFunction
    \end{algorithmic}
\end{algorithm}

\section{Programming with Objects}\label{sec:typesystem}
Although objects give the developers a way to expose dependencies and parallelism, we need develop a few object-oriented processes to make it work seamlessly with low programming overhead.

\sysname relies on the correct implementation of the MoveVM~\cite{move} and its semantics to guarantee a number of properties:
\begin{itemize}[leftmargin=*]
    \item The type checker ensures an object may only be wrapped into another object if the wrapped object is an input of the transaction, or itself indirectly a child of an input object in the transaction.
    \item The type checker ensures that an input object may only become a child of another object if the new parent object is an input of the transaction, or indirectly a child of an input to the transaction.
\end{itemize}
These checks ensure that wrapped and child objects will always have a version number that is not larger than the object that wraps them, the top-level object that wraps their parent, or their parent if it is already top-level. In turn, this ensures that when unwrapped or a child object becomes a top-level object again the Lamport timestamp corresponding to its new version number will be greater than any version number it was assigned before.
Lamport timestamps~\cite{lamport2019time} allow \sysname to efficiently delete and re-create objects without opening itself to replay attacks through version number re-use. In a nutshell, re-created objects always have a higher version number.

\para{Object Wrapping/Unwrapping} \sysname allows an executed transaction to \emph{wrap} an input object into another object that is output. The wrapped object then `disappears' from the system tables as it is in-lined within the
data of the object that wrap it. The Move type system that bounds all executions ensures that the type of a wrapped object does not change; and that it is not cloned with the same object ID.
Eventually, the wrapped object can be \emph{unwrapped} by a transactions taking the object that wraps it as an input. Once unwrapped the object is again included and tracked in the system tables
under a version that is the Lamport timestamp of the transaction that unwrapped it. This ensures that the version is greater than any previous version, preventing replay attacks on the consistent
broadcast protocol.

\para{Parent-child object relations} An owned object (called a \emph{child}) may be owned by another object (called the \emph{parent}). Parent object may be shared objects or owned objects, and the latter may themselves be children
of other parents, forming a chain of parent-child relations. A transaction is authorized to use a child object if the ultimate parent, ie.\ the object in the parent-child chain that is either shared
of owned by an address, is included in the transaction as an input. Assigning an owned object to a parent requires both the child and the parent to be included or authorized in the inputs of the
transaction that performs the assignment. This ensures that all operations on a child (including assignment) result in the full chain being updated as part of the transaction output, and assigned
the Lamport timestamp of the transaction.

Notably, children may be used as inputs to transactions implicitly without specifying them in the transactions inputs. A Move call may load a child object by object ID, and the owned objects table
(which is consistent) is used to determine the correct version of the loaded child object to use. Since the transaction executing has a lock on the parent it also has implicitly a lock on all direct
or indirect children removing any ambiguity about the correct version to use. The full chain of a dynamically loaded child is treated as if it was an input to the transaction and version numbers of
all objects on the chain updated in the output. This ensures that when the child object is assigned to an address (stops being a child) it will always have a fresh version.

\para{Objects Deletion} \label{sec:objects-deletion}
Any type of \sysname object (read-only, owned, or shared) can be either created by a transaction containing only owned objects or by a transaction containing a mix of owned and shared objects. As a result, owned objects can be created by transactions sequenced by the consensus engine, and shared objects can be created by transactions forgoing consensus. This creates a challenge when deleting shared objects that are created without consensus, as the creation transaction might be replayed causing a double-spend. To resolve this issue we rely on remembering deleted objects through tombstones. We achieve this by marking deleted objects with the version number zero upon successfully execution of the certificate (i.e., when calling $\txexec{\cert}$ in \Cref{alg:process-cert}). Remember the second transactions validity check of \Cref{sec:operations} prevents transactions from using objects with version zero as inputs. Thus setting the version of an object to zero is effectively a tombstone. Since this could create infinite memory needs we also rely on the fact that certificates are valid only while the epoch is ongoing and clear all tombstones upon epoch change.

\ifdefined\cameraReady
    \section{Reproducing Experiments} \label{sec:reproduce}

We provide the orchestration scripts\footnote{\url{https://github.com/asonnino/sui/tree/sui-lutris/crates/orchestrator}} used to benchmark the codebase evaluated in this paper on AWS .

\para{Deploying a testbed}
The file `\texttildelow/.aws/credentials' should have the following content:
\begin{footnotesize}
    \begin{verbatim}
[default]
aws_access_key_id = YOUR_ACCESS_KEY_ID
aws_secret_access_key = YOUR_SECRET_ACCESS_KEY
\end{verbatim}
\end{footnotesize}
configured with account-specific AWS \emph{access key id} and \emph{secret access key}. It is advise to not specify any AWS region as the orchestration scripts need to handle multiple regions programmatically.

A file `settings.json' contains all the configuration parameters for the testbed deployment. We run the experiments of \Cref{sec:evaluation} with the following settings:

\begin{footnotesize}
    \begin{lstlisting}[language=json]
{
    "testbed_id": "${USER}-sui",
    "cloud_provider": "aws",
    "token_file": "/Users/${USER}/.aws/credentials",
    "ssh_private_key_file": "/Users/${USER}/.ssh/aws",
    "regions": [
        "us-east-1",
        "us-west-2",
        "ca-central-1",
        "eu-central-1",
        "ap-northeast-1",
        "ap-northeast-2",
        "eu-west-1",
        "eu-west-2",
        "eu-west-3",
        "eu-north-1",
        "ap-south-1",
        "ap-southeast-1",
        "ap-southeast-2"
    ],
    "specs": "m5d.8xlarge",
    "repository": {
        "url": "https://github.com/mystenlabs/sui.git",
        "commit": "sui-lutris"
    }
}
\end{lstlisting}
\end{footnotesize}

where the file `/Users/\${USER}/.ssh/aws' holds the ssh private key used to access the AWS instances.

The orchestrator binary provides various functionalities for creating, starting, stopping, and destroying instances. For instance, the following command to boots 2 instances per region (if the settings file specifies 13 regions, as shown in the example above, a total of 26 instances will be created):
\begin{footnotesize}
    \begin{verbatim}
cargo run --bin orchestrator -- \
    testbed deploy --instances 2
\end{verbatim}
\end{footnotesize}
The following command displays he current status of the testbed instances
\begin{footnotesize}
    \begin{verbatim}
cargo run --bin orchestrator testbed status
\end{verbatim}
\end{footnotesize}
Instances listed with a green number are available and ready for use and instances listed with a red number are stopped. It is necessary to boot at least one instance per load generator, one instance per validator, and one additional instance for monitoring purposes (see below).
The following commands respectively start and stop instances:
\begin{footnotesize}
    \begin{verbatim}
cargo run --bin orchestrator -- testbed start
cargo run --bin orchestrator -- testbed stop
\end{verbatim}
\end{footnotesize}
It is advised to always stop machines when unused to avoid incurring in unnecessary costs.

\para{Running Benchmarks}
Running benchmarks involves installing the specified version of the codebase on all remote machines and running one validator and one load generator per instance. For example, the following command benchmarks a committee of 100 validators (none faulty) under a constant load of 1,000 tx/s for 10 minutes (default), using 3 load generators:
\begin{footnotesize}
    \begin{verbatim}
cargo run --bin orchestrator -- benchmark \
    --committee 100 fixed-load --loads 1000 \
    --dedicated-clients 3 --faults 0 
    --benchmark-type 0
\end{verbatim}
\end{footnotesize}
The parameter `benchmark-type` is typically set to ``0'' to instruct the load generators to submit individual payment transactions. It can also be set to ``batch'' to instruct them to submit bundles of 100 payment transactions, as experimented in \Cref{fig:bundle-latency}. When benchmarking individual transactions, we select the number of load generators by ensuring that each individual load generator produces no more than 350 tx/s (as they may quickly become the bottleneck). We set the number of load generators to 40 when benchmarking bundles of 100 transactions (\Cref{fig:bundle-latency}).

\para{Monitoring}
The orchestrator provides facilities to monitor metrics. It deploys a Prometheus instance and a Grafana instance on a dedicated remote machine. Grafana is then available on the address printed on stdout when running benchmarks with the default username and password both set to \texttt{admin}. An example Grafana dashboard can be found in the file `grafana-dashboard.json'\footnote{\url{https://github.com/asonnino/sui/blob/sui-lutris/crates/orchestrator/assets/grafana-dashboard.json}}.

\para{Troubleshooting}
The main cause of troubles comes from the genesis. Prior to the benchmark phase, each load generator creates a large number of gas object later used to pay for the benchmark transactions. This operation may fail if there are not enough genesis gas objects to subdivide or if the total system gas limit is exceeded. As a result, it may be helpful to increase the number of genesis gas objects per validator in the `genesis\_config' file\footnote{\url{https://github.com/asonnino/sui/blob/7f3d922432b185e6977513ea577929ea06097102/crates/sui-swarm-config/src/genesis\_config.rs\#L361}} when running with very small committee sizes (such as 10).

\fi

\end{document}